\documentclass[a4paper,11pt]{fullverllncs}
\usepackage[left=2.5cm,top=2.5cm,right=2.5cm,centering]{geometry}
\usepackage[dvipdfmx]{graphicx}
\usepackage{amsmath,amssymb}
\usepackage{ascmac}
\usepackage{array}
\usepackage{algpseudocode}
\usepackage{amsfonts}
\usepackage{amssymb}
\usepackage{amsmath}
\usepackage{algorithm, algpseudocode}
\usepackage{multirow}
\usepackage{arydshln}
\usepackage{hhline} 
\usepackage[misc,geometry]{ifsym}

\usepackage[colorlinks]{hyperref, xcolor}
\definecolor{winered}{rgb}{0.5,0,0}
\definecolor{darkblue}{rgb}{0,0,0.5}
\definecolor{darkgreen}{rgb}{0,0.3,0}
\hypersetup{
linkcolor=winered,
citecolor=darkblue,
urlcolor=darkgreen
}
\newcommand{\A}{\mathsf{A}}
\newcommand{\B}{\mathsf{B}}
\newcommand{\C}{\mathsf{C}}
\newcommand{\R}{\mathsf{R}}

\newcommand{\negl}{\mathsf{negl}}
\newcommand{\poly}{\mathsf{poly}}

\newcommand{\Adv}{\mathsf{Adv}}

\newcommand{\K}{\mathsf{K}}
\newcommand{\Z}{\mathbb{Z}}
\newcommand{\G}{\mathbb{G}}
\newcommand{\N}{\mathbb{N}}

\newcommand{\accept}{\mathtt{accept}}
\newcommand{\reject}{\mathtt{reject}}

\newcommand{\sfGame}{\mathsf{G}}

\newcommand{\BG}{\mathsf{BG}}
\newcommand{\BGcal}{\mathcal{BG}}

\newcommand{\pp}{\mathsf{pp}}
\newcommand{\sk}{\mathsf{sk}}
\newcommand{\pk}{\mathsf{pk}}
\newcommand{\apk}{\mathsf{apk}}
\newcommand{\EUFCMA}{\mathsf{EUF \mathchar`- CMA}}
\newcommand{\EUFCMACK}{\mathsf{EUF \mathchar`- CMA\mathchar`-CertKey}}
\newcommand{\Sign}{\mathsf{Sign}}
\newcommand{\Bad}{\mathsf{Bad}}

\newcommand{\DS}{\mathsf{DS}}
\newcommand{\nDSSetup}{\mathsf{Setup}}
\newcommand{\nDSKGen}{\mathsf{KGen}}
\newcommand{\nDSSign}{\mathsf{Sign}}
\newcommand{\nDSVerify}{\mathsf{Verify}}

\newcommand{\OMS}{\mathsf{OMS}}
\newcommand{\nOMSSetup}{\mathsf{Setup}}
\newcommand{\nOMSKGen}{\mathsf{KGen}}
\newcommand{\nOMSKVerify}{\mathsf{KVerify}}
\newcommand{\nOMSKAgg}{\mathsf{KAgg}}
\newcommand{\nOMSSign}{\mathsf{Sign}}
\newcommand{\nOMSSVerify}{\mathsf{SVerify}}

\newcommand{\AS}{\mathsf{AS}}
\newcommand{\SAS}{\mathsf{SAS}}
\newcommand{\nSASSetup}{\mathsf{Setup}}
\newcommand{\nSASKGen}{\mathsf{KGen}}
\newcommand{\nSASKVerify}{\mathsf{KVerify}}
\newcommand{\nSASSign}{\mathsf{Sign}}
\newcommand{\nSASSVerify}{\mathsf{SVerify}}

\newcommand{\Cert}{\mathsf{Cert}}

\newcommand{\DDH}{\mathsf{DDH}}
\newcommand{\DBP}{\mathsf{DBP}}

\newcommand{\BGLS}{\mathsf{BGLS03}}
\newcommand{\BGOY}{\mathsf{BGOY07}}
\newcommand{\CK}{\mathsf{CK20}}
\newcommand{\LLYt}{\mathsf{LLY13}}
\newcommand{\LLYfa}{\mathsf{LLY15a}}
\newcommand{\LLYfb}{\mathsf{LLY15b}}
\newcommand{\LOSSW}{\mathsf{LOSSW06}}
\newcommand{\McDa}{\mathsf{McD20a}}
\newcommand{\McDb}{\mathsf{McD20b}}
\newcommand{\Ours}{\mathsf{Ours}}
\newcommand{\OursCK}{\mathsf{Ours}}
\newcommand{\PS}{\mathsf{PS16}}
\newcommand{\Sch}{\mathsf{Sch11}}
\newcommand{\YCMO}{\mathsf{YCMO14}}
\newcommand{\YMOt}{\mathsf{YMO13}}

\spnewtheorem{assumption}{Assumption}{\bfseries}{\itshape}

\begin{document}

\title{Ordered Multi-Signatures with Public-Key Aggregation from SXDH Assumption\thanks{A preliminary version of this paper is appeared in the 20th International Workshop on Security (IWSEC 2025).}}
\author{Masayuki Tezuka\textsuperscript{(\Letter)} \and Keisuke Tanaka}
\authorrunning{M.Tezuka et al.}
\titlerunning{OMS with Public-Key Aggregation from SXDH Assumption}
\institute{Institute of Science Tokyo, Tokyo, Japan\\
\email{tezuka.m.eab3@m.isct.ac.jp}
}

\maketitle
\pagestyle{plain}
\noindent
\makebox[\linewidth]{September 22, 2025}

\begin{abstract}
An ordered multi-signature scheme allows multiple signers to sign a common message in a sequential manner and allows anyone to verify the signing order of signers with a public-key list.
In this work, we propose an ordered multi-signature scheme by modifying the sequential aggregate signature scheme by Chatterjee and Kabaleeshwaran (ACISP 2020).
Our scheme offers compact public parameter size and the public-key aggregation property.
This property allows us to compress a public-key list into a short aggregated key.
We prove the security of our scheme under the symmetric external Diffie-Hellman (SXDH) assumption without the random oracle model.

\end{abstract}

\section{Introduction}
\subsection{Background}\label{IntroBack}

\paragraph{\bf Aggregate Signatures (AS).}
An aggregate signature scheme introduced by Boneh, Gentry, Lynn, and Shacham \cite{BGLS03} is a special type of signature scheme that allows anyone to compress $n$ signatures produced by different signers on different messages into a short aggregate signature.
This signature scheme is meaningful when the size of an aggregate signature is independent of $n$.
This attractive feature is useful for reducing the storage space for signatures and realizing efficient verification of signatures.

Constructing aggregate signature schemes under the standard model without the random oracle model (ROM) is a difficult task.
In previous works, a multi-linear map based scheme \cite{HSW13} and an indistinguishability obfuscation (iO) based scheme~\cite{HKW15} are proposed. 
However, these schemes rely on strong assumptions.

\paragraph{\bf Sequential Aggregate Signatures (SAS).}
Due to the difficulty of constructing efficient aggregate signature schemes from standard assumptions without the ROM, several variants of aggregate signature schemes with restricted aggregation were proposed.
The one variant of the aggregate signature scheme is the sequential aggregate signature scheme proposed by Lysyanskaya, Micali, Reyzin, and Shacham~\cite{LMRS04}.
In this scheme, signatures are aggregated in a sequential manner: 
Each signer in turn sequentially signs a message and updates the signature.

\paragraph{\bf SAS Based on SXDH Assumption.} 
In previous works, several pairing-based sequential aggregate signature schemes have been proposed. \cite{CK20,LLY13,LLY15,LOSSW06,McD20,PS16,Sch11}.
We summarize these schemes in Fig.~\ref{SAS_List}.

\begin{figure}[htbp]
\begin{center}
\scalebox{1}{
\begingroup
\renewcommand{\arraystretch}{1}
\tabcolsep = 2.0pt
\begin{tabular}{lccccccc}\hline
Scheme
&
Assumption
&
$\mathcal{M}$
&
$|\pp|$
&
$|\pk|$
&
$|\sigma|$
&
Ver $\# \mathbb{P}$ \\ 
\hline
\begin{tabular}{l}
$\AS_{\BGLS}$  \\
\cite{BGLS03} $\S 3.1$ 
\end{tabular}
&
\begin{tabular}{l}
co-CDH \\
\textcolor{red}{+ROM}
\end{tabular}
&
\begin{tabular}{l}
$\{0, 1\}^{*}$
\end{tabular}
& 
\begin{tabular}{l}
$|\BGcal_{2}|+ |\G|$\\
$+ |\widetilde{\G}| + |H|$
\end{tabular}
&
\begin{tabular}{l}
$|\G|$
\end{tabular}
&
\begin{tabular}{l}
$|\widetilde{\G}|$
\end{tabular}
&
\begin{tabular}{l}
\textcolor{red}{$n+1$}
\end{tabular}
\\

\hline
\begin{tabular}{l}
$\SAS_{\LOSSW}$ \\
\cite{LOSSW06} \S 3.2
\end{tabular}
&
\begin{tabular}{l}
CDH
\end{tabular}
&
\begin{tabular}{l}
$\{0, 1\}^{\ell}$
\end{tabular}
&
\begin{tabular}{l}
$|\BGcal_{1}|$
\end{tabular}
&
\begin{tabular}{l}
\textcolor{red}{$(\ell + 1)|\G|$}\\
$+ |\G_{T}|$
\end{tabular}
&
\begin{tabular}{l}
$2|\G|$
\end{tabular}
&
\begin{tabular}{l}
$3$
\end{tabular}
\\

\hline
\begin{tabular}{l}
$\SAS_{\Sch}$ \\
\cite{Sch11} \S 3.4 
\end{tabular}
&
\begin{tabular}{l}
\textcolor{red}{LRSW${}^\dagger$}
\end{tabular}
&
\begin{tabular}{l}
$\Z^{*}_{p}$
\end{tabular}
&
\begin{tabular}{l}
$|\BGcal_{1}| + |\G|$
\end{tabular}
&
\begin{tabular}{l}
$2|\G|$
\end{tabular}
&
\begin{tabular}{l}
$4|\G|$
\end{tabular}
&
\begin{tabular}{l}
\textcolor{red}{$n$}
\end{tabular}
\\

\hline
\begin{tabular}{l}
$\SAS_{\LLYt}$ \\
\cite{LLY13} \S 3.2 
\end{tabular}
&
\begin{tabular}{l}
\textcolor{red}{LRSW${}^\dagger$}
\end{tabular}
&
\begin{tabular}{l}
$\Z^{*}_{p}$
\end{tabular}
&
\begin{tabular}{l}
$|\BGcal_{1}| + 2|\G|$
\end{tabular}
&
\begin{tabular}{l}
$|\G|$
\end{tabular}
&
\begin{tabular}{l}
$3|\G|$
\end{tabular}
&
\begin{tabular}{l}
$5$
\end{tabular}
\\

\hline
\begin{tabular}{l}
$\SAS_{\LLYfa}$ \\
\cite{LLY15} \S 4.2.1
\end{tabular}
&
\begin{tabular}{l}
SXDH\textcolor{red}{+LW2${}^{\ddag}$} \\
+DBDH
\end{tabular}
&
\begin{tabular}{l}
$\Z^{*}_{p}$ 
\end{tabular}
&
\begin{tabular}{l}
$|\BGcal_{3}|$\\
$+ 5|\G| + 7|\widetilde{\G}|$
\end{tabular} 
&
\begin{tabular}{l}
$2|\G| +8|\widetilde{\G}|$\\
$+ |\G_{T}|$
\end{tabular}
&
\begin{tabular}{l}
$8|\G|$
\end{tabular}
&
\begin{tabular}{l}
$8$
\end{tabular}
\\

\hline
\begin{tabular}{l}
$\SAS_{\LLYfb}$ \\
\cite{LLY15} \S 4.2.2 
\end{tabular}
&
\begin{tabular}{l}
\textcolor{red}{LW1${}^{\ddag}$+LW2${}^{\ddag}$}\\
+DBDH 
\end{tabular}
&
\begin{tabular}{l}
$\Z^{*}_{p}$ 
\end{tabular}
&
\begin{tabular}{l}
$|\BGcal_{3}| + 6|\G|$\\
$ + 3|\widetilde{\G}| +  |\G_{T}|$
\end{tabular}
&
\begin{tabular}{l}
$6|\G| +6|\widetilde{\G}|$\\
$ + |\G_{T}|$
\end{tabular}
&
\begin{tabular}{l}
$6|\G|$
\end{tabular}
&
\begin{tabular}{l}
$6$
\end{tabular}
\\

\hline
\begin{tabular}{l}
$\SAS_{\PS}$\\
\cite{PS16} \S 5 
\end{tabular}
&
\begin{tabular}{l} 
\textcolor{red}{PS${}^\dagger$}
\end{tabular}
&
\begin{tabular}{l}
$\Z^{*}_{p}$
\end{tabular}
&
\begin{tabular}{l}
$|\BGcal_{3}|$\\
$ +  2|\G| + 2|\widetilde{\G}|$
\end{tabular}
&
\begin{tabular}{l}
$|\widetilde{\G}|$
\end{tabular}
&
\begin{tabular}{l}
$2|\G|$
\end{tabular}
&
\begin{tabular}{l}
$2$
\end{tabular}
\\

\hline
\begin{tabular}{l}
$\SAS_{\McDa}$ \\
\cite{McD20} \S 5.3
\end{tabular}
&
\begin{tabular}{l}
\textcolor{red}{$\ell$-PS${}^\dagger$}
\end{tabular}
&
\begin{tabular}{l}
$(\Z^{*}_{p})^{\ell}$ 
\end{tabular}
&
\begin{tabular}{l}
$|\BGcal_{3}|$\\
$+ 2|\G| + 2|\widetilde{\G}|$ 
\end{tabular}
&
\begin{tabular}{l}
$\ell|\widetilde{\G}|$
\end{tabular}
&
\begin{tabular}{l}
$2|\G|$
\end{tabular}
&
\begin{tabular}{l}
$2$
\end{tabular}
\\

\hline
\begin{tabular}{l}
$\SAS_{\McDb}$ \\
\cite{McD20} \S 5.4 
\end{tabular}
&
\begin{tabular}{l}
\textcolor{red}{$\ell$-PolyS${}^\dagger$}
\end{tabular}
&
\begin{tabular}{l}
$(\Z^{*}_{p})^{\ell}$ 
\end{tabular}
&
\begin{tabular}{l}
$|\BGcal_{3}|$\\
$+ 2|\G| + 2|\widetilde{\G}|$ 
\end{tabular}
&
\begin{tabular}{l}
$\ell|\widetilde{\G}|$
\end{tabular}
&
\begin{tabular}{l}
$2|\G|$
\end{tabular}
&
\begin{tabular}{l}
$2$
\end{tabular}
\\

\hline
\begin{tabular}{l}
$\SAS_{\CK}$ \\
\cite{CK20} \S 4.1
\end{tabular}
&
\begin{tabular}{l}
SXDH
\end{tabular}
&
\begin{tabular}{l}
$\Z^{*}_{p}$
\end{tabular}
&
\begin{tabular}{l}
$|\BGcal_{3}|$\\
$+ 4|\G| + 3|\widetilde{\G}|$
\end{tabular}
&
\begin{tabular}{l}
$|\widetilde{\G}|$
\end{tabular}
&
\begin{tabular}{l}
$3|\G|$
\end{tabular}
&
\begin{tabular}{l}
$3$
\end{tabular}
\\
\hline
\end{tabular}
\endgroup
}
\end{center}
\caption{\small Comparison among pairing-based sequential aggregate signature schemes.
We highlight the weakness of the corresponding scheme compared to $\SAS_{\CK}$ in \textcolor{red}{red}.
In the column ``Assumption'', $\dagger$ represents interactive assumptions and $\ddag$ represents non-standard static assumptions.
The column ``$\mathcal{M}$'' represents a message space.
The column ``$|\pp|$'' (resp. ``$|\pk|$'', ``$|\sigma|$'')  represents the size of public parameters, (resp. public key, signature).
In the column ``$|\pp|$'',  $|\BGcal_{i}|$ (resp. $|H|$) is the size of the description of a type $i$ pairing group (resp. hash function $H$).
The column of ``Ver $\# \mathbb{P}$'' represents the number of pairing operations in the signature verification and $n$ is the number of signatures compressed into an aggregate signature.}
\label{SAS_List}
\end{figure}

$\SAS_{\CK}$ has following strengths compared with schemes in Fig.~\ref{SAS_List}:
(1)Schemes $\AS_{\BGLS}$, $\SAS_{\Sch}$ \cite{Sch11}, $\SAS_{\LLYt}$ \cite{LLY13}, $\SAS_{\LLYfa}$ \cite{LLY15}, $\SAS_{\LLYfb}$ \cite{LLY15}, $\SAS_{\PS}$ \cite{PS16}, $\SAS_{\McDa}$ \cite{McD20}, and $\SAS_{\McDb}$ rely on interactive assumptions (the LRSW assumption \cite{LRSW99}, the PS~\cite{PS16} assumption, the $\ell$-PS assumption \cite{McD20}, and the $\ell$-polyS assumption \cite{McD20}), non-standard static assumptions (the LW1 assumption \cite{LW10} and the LW2 assumption \cite{LW10}), or the ROM. The security of $\SAS_{\CK}$ is proven under the symmetric external Diffie-Hellman (SXDH) assumption \cite{BGMM05} without the ROM.
(2)In $\SAS_{\LOSSW}$ \cite{LOSSW06}, the number of group elements in a public key $\pk$ depends on the message length.
In $\SAS_{\CK}$, the number of group elements in $\pk$ is independent of the message length.
(3)In the signature verification of  $\AS_{\BGLS}$ and $\SAS_{\Sch}$, the number of pairing operations depends on the number of signatures compressed into an aggregate signature.
In the signature verification of $\SAS_{\CK}$, the number of pairing operations is the constant $3$.

\paragraph{\bf Multi-Signatures (MS) with Public-Key Aggregation.}
A multi-signature scheme introduced by Itakura and Nakamura \cite{IN83} is an interactive protocol that enables $n$ signers collaboratively to sign a common message.
We verify a multi-signature with the set of public keys which correspond to the signers.
The size of the multi-signature should be independent of $n$.

Maxwell, Poelstra, Seurin, and Wuille \cite{MPSW19} proposed a multi-signature scheme with public-key aggregation.
This property allows us to compress a set of public keys into a compact aggregated public key.
This feature is useful for shrinking the transaction data associated with Bitcoin Multisig addresses.
Several multi-signature schemes with constructions with public-key aggregation (e.g., \cite{BDN18,FH21,MPSW19}) were proposed in the ROM.

\paragraph{\bf Ordered Multi-Signatures (OMS).}
A drawback of a multi-signature scheme is that it requires interactions among the signers who collaboratively sign a common message to generate a signature.
Boldyreva, Gentry, O'Neill, and Yum \cite{BGOY07} proposed an ordered multi-signature scheme that allows signers to sign a common message in a sequential manner.
A sequentially generated signature has a property that the signing order of signers can be verified.

This feature is useful for a network routing application.
In the network routing application, we want to verify the path (i.e., the ordered list of routers)  of the packet that travels to reach its destination.
By using an ordered multi-signature scheme to have routers sequentially sign a packet passed through it, we can verify the traveled path of the packet.
Ordered multi-signatures are also useful for proofs of sequential communication delay \cite{BDPT24}.

\subsection{Motivation and Technical Problem}\label{Subsec_Moti_Pro}
\paragraph{\bf Motivation: OMS with Public-Key Aggregation.}
To make ordered multi-signatures practical, it is desired to construct an ordered multi-signature scheme with public-key aggregation under the standard assumption.
Moreover, a scheme with a smaller public parameter size, public key size, and efficient signature verification is desirable.

To mitigate the public key size problem, the public-key aggregation property is useful.
Let us consider the following application scenario of ordered multi-signatures for network routing.
Suppose that network routers sequentially sign a packet through them and confirm that a packet passed through a particular path by verifying an ordered multi-signature.
Consider the situation that a large number of packets pass through a particular path and we want to focus on verifying packets for the specific path.
In this case, if the values of some calculations for signature verifications are common, these values can be reused in signature verifications and the total time of many signatures in verification can be reduced.
The public-key aggregation property allows to reuse of aggregated public keys to verify signatures for the common network pass. 
However, as far as we know, this property has not been considered in previous works of ordered multi-signatures.

\paragraph{\bf Technical Problem: OMS from $\SAS_{\CK}$ Does Not Support Public-Key Aggregation Property.}
Boldyreva et al. \cite{BGOY07} explained how to transform an ordered multi-signature scheme from a sequential aggregate signature scheme.
Thanks to their transformation, we obtain the efficient ordered multi-signature scheme from the efficient sequential aggregate signature $\SAS_{\CK}$.
However, this scheme does not have the public-key aggregation property. 
Here, we briefly explain why the derived scheme does not support the public-key aggregation.

Let $(p, \G, \widetilde{\G}, \G_{T}, e)$ be a bilinear group and $\widetilde{U}, \widetilde{D}, \widetilde{H}$ be a group elements which are included in public parameter.
The public key of the derived scheme is $\pk_{i} = \widetilde{V}_{i}$ where $\widetilde{V}_{i}$ is a group element of $\widetilde{\G}$.
The signature on a message $m$ signed by signers 1 to $n$ consists of a tuple $\sigma = (A, B, C)$ where $A, B, C$ are group elements of $\widetilde{\G}$.
In the verification of ordered multi-signature $\sigma = (A, B, C)$ on a public key list $L=(\pk_{1}, .., \pk_{n})$ and a message $m$, we checks the pairing equation $e(A, \widetilde{U}\prod_{i \in [n]} \widetilde{V}^{m||i}_{i}) \cdot e(B, \widetilde{D}) = e(C, \widetilde{H})$.
If we consider the simple public-key aggregation key $\apk = \prod_{i \in [n]}\pk_{i}$, this does not allow for signature verification.
Since the elements $\widetilde{U}\prod_{i \in [n]} \widetilde{V}^{m||i}_{i}$ cannot be computed from $\apk$.
We see that the obtained ordered multi-signature scheme from the original $\SAS_{\CK}$ unlikely has the public-key aggregation.
We need some modifications to construct an ordered multi-signature signature scheme with public-key aggregation from the standard assumption without the ROM.

\subsection{Our Contributions}\label{Subsec_Our_Contri}
\paragraph{\bf  Contributions.}
First, we modify $\SAS_{\CK}$ and propose our sequential aggregate signature scheme $\SAS_{\OursCK}$ to support a vector message signing of $(\Z^{*}_{p})^{\ell}$ in Section \ref{Subsec_SAS_Ours_CK}.
In the case of $\ell = 1$, $\SAS_{\OursCK}$ corresponds to $\SAS_{\CK}$.
Next, we apply the transformation by Boldyreva et al. \cite{BGOY07} to $\SAS_{\OursCK}$ with $\ell = 2$.
The algebraic structure of the obtained ordered multi-signature scheme is nicely suited to the public-key aggregation property.
We explain detail of this fact in Section \ref{Subsec_OMS_Over_Ours_View}.
Then, we modify the derived ordered multi-signature scheme to obtain the ordered multi-signature scheme $\OMS_{\OursCK}$ in Section \ref{Subsec_OMS_Ours_CK}. 

\paragraph{\bf Comparison with OMS in Previous Works.}
We summarize pairing-based ordered multi-signature schemes in previous works \cite{BGOY07,YCMO14,YMO13} and derived ordered multi-signature schemes by applying the transformation of Boldyreva et al. \cite{BGOY07} to efficient sequential aggregate signature schemes \cite{CK20,LLY13,McD20,PS16} in Fig.~\ref{OMS_List}.

\begin{figure}[h]
\begin{center}
\scalebox{1}{
\begingroup
\renewcommand{\arraystretch}{1}
\tabcolsep = 2.0pt
\begin{tabular}{lccccccccc}\hline
Scheme
&
Assumption
&
$\mathcal{M}$
&
$|\pp|$
&
$|\pk|$
&
$|\sigma|$
&
Ver $\# \mathbb{P}$
&
$|\apk|$ \\%& CK model \\
\hline
\begin{tabular}{l}
$\OMS_{\BGOY}$\\
\cite{BGOY07} \S 3.2 
\end{tabular}
&
\begin{tabular}{l}
CDH\\
\textcolor{red}{+ROM}
\end{tabular}
&
\begin{tabular}{l}
$\{0, 1\}^{*}$
\end{tabular}
&
\begin{tabular}{l}
$|\BGcal_{1}| + |\G|$\\
$ + |\widetilde{\G}| + |H|$
\end{tabular}
&
\begin{tabular}{l}
$3|\G|$
\end{tabular}
&
\begin{tabular}{l}
$2|\G|$
\end{tabular}
&
\begin{tabular}{l}
$3$
\end{tabular}
&
\begin{tabular}{l}
$2|\G|^{\flat}$
\end{tabular}
\\

\hline
\begin{tabular}{l}
$\OMS_{\YMOt}$\\
\cite{YMO13} \S 4.2
\end{tabular}
&
\begin{tabular}{l}
CDH
\end{tabular}
&
\begin{tabular}{l}
$\{0, 1\}^{\ell}$
\end{tabular}
&
\begin{tabular}{l}
$|\BGcal_{1}|$\\
\textcolor{red}{$+ (\ell+3)|\G|$}
\end{tabular}
&
\begin{tabular}{l}
$3|\G|$
\end{tabular}
&
\begin{tabular}{l}
$3|\G|$
\end{tabular}
&
\begin{tabular}{l}
$4$
\end{tabular}
&
\begin{tabular}{l}
$2|\G|^{\flat}$
\end{tabular}
\\

\hline
\begin{tabular}{l}
$\OMS_{\YCMO}$\\
\cite{YCMO14} \S 4.2
\end{tabular}
&
\begin{tabular}{l}
CDH
\end{tabular}
&
\begin{tabular}{l}
$\{0, 1\}^{\ell}$
\end{tabular}
&
\begin{tabular}{l}
$|\BGcal_{1}|$\\
\textcolor{red}{$+ (\ell+3)|\G|$}
\end{tabular}
&
\begin{tabular}{l}
$3|\G|$
\end{tabular}
&
\begin{tabular}{l}
$2|\G|$
\end{tabular}
&
\begin{tabular}{l}
$4$
\end{tabular}
&
\begin{tabular}{l}
$2|\G|^{\flat}$
\end{tabular}
\\

\hline

\begin{tabular}{l}
$\SAS_{\LLYt}{}^\sharp$ \\
\cite{LLY13} \S 3.2 
\end{tabular}
&
\begin{tabular}{l}
\textcolor{red}{LRSW${}^\dagger$}
\end{tabular}
&
\begin{tabular}{l}
Prefix \\
of $\Z^{*}_{p}{}^{\sharp}$
\end{tabular}
&
\begin{tabular}{l}
$|\BGcal_{1}| + 2|\G|$
\end{tabular}
&
\begin{tabular}{l}
$|\G|$
\end{tabular}
&
\begin{tabular}{l}
$3|\G|$
\end{tabular}
&
\begin{tabular}{l}
$5$
\end{tabular}
&
\begin{tabular}{l}
\textcolor{red}{$\times$}
\end{tabular}
\\
\hline

\begin{tabular}{l}
$\SAS_{\PS}{}^\sharp$\\
\cite{PS16} \S 5 
\end{tabular}
&
\begin{tabular}{l} 
\textcolor{red}{PS${}^\dagger$}
\end{tabular}
&
\begin{tabular}{l}
Prefix\\ 
of $\Z^{*}_{p}{}^{\sharp}$
\end{tabular}
&
\begin{tabular}{l}
$|\BGcal_{3}|$\\
$ +  2|\G| + 2|\widetilde{\G}|$
\end{tabular}
&
\begin{tabular}{l}
$|\widetilde{\G}|$
\end{tabular}
&
\begin{tabular}{l}
$2|\G|$
\end{tabular}
&
\begin{tabular}{l}
$2$
\end{tabular}
&
\begin{tabular}{l}
\textcolor{red}{$\times$}
\end{tabular}
\\
\hline

\begin{tabular}{l}
$\SAS_{\McDa}{}^\sharp$ \\
\cite{McD20} \S 5.3
\end{tabular}
&
\begin{tabular}{l}
\textcolor{red}{$2$-PS${}^\dagger$}
\end{tabular}
&
\begin{tabular}{l}
$\Z^{*}_{p}{}^{\sharp}$ 
\end{tabular}
&
\begin{tabular}{l}
$|\BGcal_{3}|$\\
$+ 2|\G| + 2|\widetilde{\G}|$ 
\end{tabular}
&
\begin{tabular}{l}
$2|\widetilde{\G}|$
\end{tabular}
&
\begin{tabular}{l}
$2|\G|$
\end{tabular}
&
\begin{tabular}{l}
$2$
\end{tabular}
&
\begin{tabular}{l}
$2|\widetilde{\G}|^{\flat}$
\end{tabular}
\\
\hline

\begin{tabular}{l}
$\SAS_{\CK}{{}^\sharp}$ \\
\cite{CK20} \S 4.1
\end{tabular}
&
\begin{tabular}{l}
SXDH
\end{tabular}
&
\begin{tabular}{l}
Prefix \\
of $\Z^{*}_{p}{}^{\sharp}$
\end{tabular}
&
\begin{tabular}{l}
$|\BGcal_{3}|$\\
$+ 4|\G| + 3|\widetilde{\G}|$
\end{tabular}
&
\begin{tabular}{l}
$|\widetilde{\G}|$
\end{tabular}
&
\begin{tabular}{l}
$3|\G|$
\end{tabular}
&
\begin{tabular}{l}
$3$
\end{tabular}
&
\begin{tabular}{l}
\textcolor{red}{$\times$}
\end{tabular}
\\

\hline

\begin{tabular}{l}
$\OMS_{\OursCK}$\\
\S \ref{Subsec_OMS_Ours_CK}
\end{tabular}
&
\begin{tabular}{l}
SXDH
\end{tabular}
&
\begin{tabular}{l}
$\Z^{*}_{p}$
\end{tabular}
&
\begin{tabular}{l}
$|\BGcal_{3}|$\\
$+ 4|\G| + 3|\widetilde{\G}|$
\end{tabular}
&
\begin{tabular}{l}
$2|\widetilde{\G}|$
\end{tabular}
&
\begin{tabular}{l}
$3|\G|$
\end{tabular}
&
\begin{tabular}{l}
$3$
\end{tabular}
&
\begin{tabular}{l}
$2|\widetilde{\G}|$
\end{tabular}

\\

\hline
\end{tabular}
\endgroup
}
\end{center}
\caption{\small Comparison among pairing-based ordered multi-signature schemes.
We highlight the weakness of the corresponding scheme compared to our scheme $\OMS_{\Ours}$ in \textcolor{red}{red}.
In the column ``Assumption'', $\dagger$ represents interactive assumptions.
The column ``$|\pp|$'' (resp. ``$|\pk|$'', ``$|\sigma|$'', ``$|\apk|$'') represents the size of public parameters, (resp. public key, signature, aggregated public key).
In the column ``$|\pp|$'', $|\BGcal_{i}|$ (resp. $|H|$) is the size of the description of a type $i$ pairing group (resp. hash function $H$).
The column of ``Ver $\# \mathbb{P}$'' represents the number of pairing operations in the signature verification.
In the column ``$|\apk|$'', $\times$ represents that the scheme does not support public-key aggregation.
$\sharp:$We consider ordered multi-signature schemes by applying the transformation by Boldyreva et al. \cite{BGOY07} to $\SAS_{\LLYt}$,$\SAS_{\PS}$, $\SAS_{\McDa}$, and $\SAS_{\CK}$.
Due to this transformation, the message spaces of  $\SAS_{\LLYt}$,$\SAS_{\PS}$, and $\SAS_{\CK}$ are restricted to prefix of $\Z^{*}_{p}$.
For $\SAS_{\McDa}$, we consider the $\SAS_{\McDa}$ with $\ell=2$ in Fig. \ref{SAS_List}.
This derived ordered multi-signature scheme has message space $\Z^{*}_{p}$.
$\flat:$ The originally proposed paper of the corresponding scheme does not consider the public-key aggregation. With a slight modification, the scheme supports the public-key aggregation property.}
\label{OMS_List}
\end{figure}

$\OMS_{\OursCK}$ has following strengths:
(1)The security of $\OMS_{\OursCK}$ is proven under the SXDH assumption without the ROM.
$\OMS_{\BGOY}$ \cite{BGOY07}, $\SAS_{\LLYt}$ \cite{LLY13},  $\SAS_{\PS}$ \cite{PS16}, and $\SAS_{\McDa}$ \cite{McD20} rely on interactive assumptions or the ROM.
(2)The number of group elements of public parameters of $\OMS_{\OursCK}$ is independent of the message length.
By contrast, $\OMS_{\YMOt}$ \cite{YMO13} and $\OMS_{\YCMO}$ \cite{YCMO14} depends on the message length $\ell$.
(3)$\OMS_{\OursCK}$ supports the public-key aggregation property. 
But ordered multi-signature schemes from $\SAS_{\LLYt}$ \cite{LLY13},  $\SAS_{\PS}$ \cite{PS16}, and $\SAS_{\CK}$
\cite{CK20} do not have this property.
Recently, the Schnorr-based ordered multi-signature scheme was proposed in \cite{BDPT25}. This scheme can also be easily adapted to obtain public key aggregation by using the product among all public keys as aggregate public key.

\section{Preliminaries}\label{Sec_Prelimi}
In this section, we introduce notations and review the signature scheme $\DS_{\CK}$ by Chatterjee and Kabaleeshwaran  \cite{CK20}.

\subsection{Notations}
Let $1^{\lambda}$ be the security parameter. 
A function $f(\lambda)$ is negligible in $\lambda$ if $f(\lambda)$ tends to $0$ faster than $\frac{1}{\lambda^c}$ for every constant $c > 0$.
For an algorithm $\A$, $y \leftarrow \A(x)$ denotes that the algorithm $\A$ outputs $y$ on input~$x$.
We abbreviate probabilistic polynomial time as PPT.

For a positive integer $n$, we define $[n]: =\{1,\dots, n\}$.
For a finite set $S$, $s \xleftarrow{\$} S$ represents that an element $s$ is chosen from $S$ uniformly at random.
We denote a set of infinite bit strings as $\{0, 1\}^{*}$.
For a list $L$, $|L|$ represents the number of elements in $L$.
For a group $\G$, we define $\G^* := \G \backslash \{1_{\G}\}$ where $1_{\G}$ is the identity element of $\G$.
In this work, we consider type-3 pairing groups.
$\BGcal_{3}= (p, \G, \widetilde{\G}, \G_T, e)$ represents a description of a type-3 pairing group. 
We represent a bilinear group generator as $\BG$.
We supply the definition of a bilinear group, a bilinear group generator, and hardness assumptions in Appendix \ref{Appen_BG}.

\subsection{Scheme $\DS_{\CK}$ \cite{CK20}}\label{Subsec_DSCK}
We give the definition of a digital signature scheme in Appendix \ref{Appen_DS_Def}.
Here, we review the signature scheme $\DS_{\CK}$ by Chatterjee and Kabaleeshwaran  \cite{CK20}.
$\DS_{\CK}$ is a randomizable signature scheme.
This scheme has a feature that a signature $\sigma$ on a message $m$ is refreshed to $\sigma'$ without the signing key.
We cannot distinguish whether $\sigma'$ is directly output by the signing algorithm or refreshed.
The message space $\mathcal{M}$ of $\DS_{\CK}$ is $(\Z_{p})^{\ell}$.
The construction of $\DS_{\CK}$ is given in Fig.\ref{DSCKconst}.

\begin{figure}[htbp]
\centering
\begin{tabular}{|l|}
\hline
$\nDSSetup(1^{\lambda}):$\\
~~~$\BGcal_{3}= (p, \G, \widetilde{\G}, \G_T, e) \leftarrow \BG(1^\lambda)$, $G \xleftarrow{\$} \G^*$, $\widetilde{G} \xleftarrow{\$} \widetilde{\G}^*$\\
~~~Return $\pp \leftarrow (\BGcal_{3}, G, \widetilde{G})$.\\

$\nDSKGen (\pp):$\\
~~~$d, x_{1},  x_{2},  \xleftarrow{\$} \mathbb{Z}^{*}_{p}$,  $(y_{j,1},  y_{j,2})_{j \in [\ell]}  \xleftarrow{\$} (\mathbb{Z}^{*}_{p})^{2\ell}$, $\widetilde{H} \xleftarrow{\$} \widetilde{\G}$, $\widetilde{D} \leftarrow \widetilde{H}^{d}$,\\
~~~$\widetilde{U} \leftarrow \widetilde{H}^{x_{2} - dx_{1}}$, $(\widetilde{V}_{j} \leftarrow \widetilde{H}^{y_{j, 2} - dy_{j, 1}})_{j \in [\ell]}$.\\
~~~Return $(\pk, \sk) \leftarrow ((\widetilde{H}, \widetilde{D}, \widetilde{U}, (\widetilde{V}_{j})_{j \in [\ell]}), (x_{1},  x_{2}, (y_{j,1},  y_{j, 2})_{j \in [\ell]}))$.\\

$\nDSSign(\sk=(x_{1},  x_{2}, (y_{j,1},  y_{j, 2})_{j \in [\ell]}), m=(m_{j})_{j \in [\ell]}):$\\
~~~$r \xleftarrow{\$} \mathbb{Z}^{*}_{p}$, $A \leftarrow G^{r}$, $B \leftarrow A^{x_{1} + \sum_{j \in [\ell]}m_{j}y_{j, 1}}$, $C \leftarrow A^{x_{2} + \sum_{j \in [\ell]}m_{j}y_{j, 2}}$,\\
~~~Return $\sigma = (A, B, C)$.\\

$\nDSVerify (\pk=(\widetilde{H}, \widetilde{D}, \widetilde{U}, (\widetilde{V}_{j})_{j \in [\ell]}), m=(m_{j})_{j \in [\ell]}, \sigma = (A, B, C)):$\\
~~~If $A = 1_{\G}$, return $0$.\\
~~~If $e(A, \widetilde{U} \prod_{j \in [\ell]}\widetilde{V}_{j}^{m_{j}} ) \cdot e(B, \widetilde{D}) = e(C, \widetilde{H})$, return $1$.\\
~~~Otherwise return $0$.\\
\hline
\end{tabular}
\caption{\small
The signature scheme $\DS_{\CK}$ \cite{CK20}.}
\label{DSCKconst}
\end{figure}

\begin{lemma}[\cite{CK20}]\label{Th_DSCK_EUF_from_SXDH}
If the SXDH assumption holds (See Appendix \ref{Appen_BG} for the SXDH assumption), then $\DS_{\CK}$ satisfies the EUF-CMA security.
\end{lemma}

\section{Sequential Aggregate Signatures (SAS)}\label{Sec_SAS}
In this section, first, we review a definition of a sequential aggregate signature scheme and its security notion.
Then, we present our sequential aggregate signature scheme $\SAS_{\OursCK}$ which is an extension scheme of $\SAS_{\CK}$.
Finally, we prove the security of $\SAS_{\OursCK}$.

\subsection{Sequential Aggregate Signature Scheme $\SAS$}
We review a definition of a sequential aggregate signature scheme and its security notion.
\begin{definition}[Sequential Aggregate Signature Scheme]
A sequential aggregate signature scheme $\SAS$ consists of the following tuple of algorithms $(\nSASSetup, \allowbreak \nSASKGen, \nSASKVerify, \allowbreak \nSASSign, \nSASSVerify)$.
\begin{itemize}
\item $\nSASSetup (1^{\lambda}):$ A setup algorithm takes as an input a security parameter $1^{\lambda}$. It returns the public parameter $\pp$.
In this work, we assume that $\pp$ defines a message space and represents this space by $\mathcal{M}$.
We omit a public parameter $\pp$ in the input of all algorithms except for $\nOMSKGen$.

\item $\nSASKGen (\pp):$ A key generation algorithm takes as an input a public parameter $\pp$. It returns a public key $\pk$ and a secret key $\sk$.

\item $\nSASKVerify(\pk, \sk):$ A key verification algorithm takes as an input a public key $\pk$ and a secret key $\sk$. It returns a bit $b \in  \{0, 1\}$.

\item $\nSASSign(\sk_{n}, L_{n-1} = (\pk_{1}, \dots, \pk_{n-1}), (m_{1}, \dots,  m_{n-1}), m_{n}, \sigma_{n-1}):$ A signing algorithm takes as an input a secret key $\sk_{n}$, a list of public keys $L_{n-1}=(\pk_{1}, \dots, \pk_{n-1})$, a list of messages $(m_{1}, \dots,  m_{n-1})$, a message $m_{n}$, and a signature $\sigma_{n-1}$. 
It returns an updated signature $\sigma_{n}$ or $\bot$.

\item $\nSASSVerify (L_{n} = (\pk_{1}, \dots, \pk_{n}), (m_{1}, \dots, m_{n}), \sigma):$ A signature verification algorithm takes as an input a list of public key $L$, a list of message  $m_{i}$, and a signature $\sigma$.
It returns a bit $b \in  \{0, 1\}$.
\end{itemize}
\end{definition}

\paragraph{\bf Correctness.}
$\SAS$ satisfies correctness if $\forall \lambda, n \in \N$, $\pp \leftarrow \nSASSetup (1^{\lambda})$, $\forall m_{i} \in \mathcal{M}$ for $i \in [n]$, $(\pk_{i}, \sk_{i}) \leftarrow \nOMSKGen(\pp)$ for $i \in [n]$, $L_{0} = \epsilon$, $\sigma_{0} = \epsilon$, $L_{i}=(\pk_{1}, \dots, \pk_{i})$ for $i \in [n]$, and $\sigma_{i} \leftarrow \nOMSSign(\sk_{i}, L_{i-1}, (m_{1}, \dots, m_{i-1}), m_{i}, \sigma_{i-1})$ for $i \in [n]$, 
\begin{enumerate}
\item For $i \in [n]$, $\nSASKVerify(\pk_{i}, \sk_{i}) = 1$
\item For $i \in [n]$, if elements in $L_{i}$ are distinct, $\nSASSVerify(L_{i}, (m_{1}, \dots, m_{i}), \sigma_{i}) = 1$
\end{enumerate}
holds.

We review a security notion for a sequential aggregate signature scheme.
In this work, we consider the existentially unforgeable under chosen message attacks (EUF-CMA) security in the certified key model \cite{LOSSW06,LLY13}.
This security guarantees that it is hard for any PPT adversary to forge an aggregate signature on a set of messages for a set of signers whose secret keys are not all known to the adversary.
In the certified key model, the adversary is allowed to generate $(\pk, \sk)$ for any non-target signer and use this key to generate a forgery. 
But the adversary must register $\pk$ by proving the knowledge of $\sk$ via the key registration query.

\begin{definition}[EUF-CMA Security in CK model]\label{Def_SAS_EUF_CK}
Let $\SAS$ be a sequential aggregate signature scheme and $\A$ be a PPT adversary.
The existentially unforgeable under chosen message attacks (EUF-CMA) security in the certified key model is defined by the following game $\sfGame^{\EUFCMACK}_{\SAS, \A}(1^{\lambda})$ between the challenger $\C$ and an adversary $\A$.

\begin{itemize}
\item {\bf Initial setup:}
$\C$ initializes sets $S^{\Cert} \leftarrow \{\}$, $S^{\Sign} \leftarrow \{\}$, runs $\pp \leftarrow \nSASSetup (1^{\lambda})$, $(\pk^{*}, \sk^{*}) \leftarrow \nSASKGen(\pp)$, and sends $(\pp, \pk^{*})$ to $\A$.
\item $\A$ makes key registration queries and signing queries polynomial many times.
\begin{itemize}
\item {\bf Key registration query:}
For a key registration query on $(\pk, \sk)$, $\C$ checks the validity of $(\pk, \sk)$.
If $\nSASKVerify(\pk, \sk) = 1$, $\C$ updates $S^{\Cert} \leftarrow S^{\Cert} \cup \{\pk\}$ and returns $\accept$.
Otherwise, $\C$ returns $\reject$.

\item {\bf Signing query:}
For an signing query on $(L_{n-1}=(\pk_{1}, \dots, \pk_{n-1}), (m_{1}, \allowbreak \dots, m_{n-1}), m_{n},  \allowbreak \sigma_{n-1})$, $\C$ proceeds as follows. \\
If $L_{n-1} \neq  \epsilon$ (i.e., $n-1 \geq 1$), $\C$ checks that the following conditions:
\begin{enumerate}
\item $\pk_{i} \in Q^{\Cert}$ for $i \in [n-1] ;$
\item $\nSASSVerify(L_{n-1}, (m_{1}, \dots, m_{n-1}), \sigma_{n-1}) = 1$.
\end{enumerate}
If at least one of the above conditions does not hold, $\mathcal{O}^{\Sign}$ returns $\bot$.
Then $\C$ runs $\sigma_{n} \leftarrow \nSASSign(\sk^{*}, L_{n-1}, (m_{1}, \dots,  m_{n-1}), m_{n}, \sigma_{n-1})$ and updates $S^{\Sign} \leftarrow S^{\Sign} \cup \{ m_{n} \}$, and returns $\sigma_{n}$ to $\A$.
\end{itemize}
\item  {\bf End of the game:} 
$\A$ finally outputs a forgery $(L^{*}_{n^{*}}=(\pk^{*}_{1}, \dots, \pk^{*}_{n^{*}}),  (m_{1}^{*},  \allowbreak\dots, m^{*}_{n^{*}}), \sigma^{*}_{n^{*}})$ to $\C$.
If the following conditions hold, return $1$.
\begin{enumerate}
\item $\nSASSVerify(L^{*}_{n^{*}}, (m_{1}^{*}, \dots, m^{*}_{n^{*}}), \sigma^{*}_{n^{*}}) = 1 ;$
\item There exists $i^{*} \in [n^{*}]$ such that $\pk^{*}_{i^{*}} = \pk^{*} \land m^{*}_{i^{*}} \notin S^{\Sign} ;$
\item $\forall j \in [n^{*}]$ such that $\pk_{j} \neq \pk^{*}$, $\pk_{j} \in S^{\Cert}.$
\end{enumerate}
\end{itemize}

The advantage of $\A$ is defined by $\Adv^{\EUFCMACK}_{\SAS, \A}(\lambda) \allowbreak:=$ $\Pr[\sfGame^{\EUFCMACK}_{\SAS, \A}(1^{\lambda}) \Rightarrow 1]$.
$\SAS$ satisfies the EUF-CMA security in the CK model if for any PPT adversary $\A$, $\Adv^{\EUFCMACK}_{\SAS, \A}(\lambda)$ is negligible in $\lambda$.
\end{definition}

\subsection{Schemes $\SAS_{\CK}$ \cite{CK20} and $\SAS_{\OursCK}$}\label{Subsec_SAS_Ours_CK}

In this section, we give our sequential aggregate signature scheme $\SAS_{\OursCK}$ by modifying $\SAS_{\CK}$.
Commonly, both $\SAS_{\CK}$ \cite{CK20} and $\SAS_{\OursCK}$ are obtained by applying the public-key sharing technique \cite{LLY13} to $\DS_{\CK}$.
This technique allows us to construct a sequential aggregate signature scheme from a digital signature scheme with randomizable property \cite{LLY13,PS16,CK20}.
This technique is applied to a randomizable signature scheme as follows.
First, divide elements in the public key of the original signature scheme into shared (common) elements and non-shared elements.
Then, add shared elements into the public parameter and modify a key generation of the original signature to generate only non-shared elements.
By this modification, we force signers to use the same elements in public parameters in the signing.

Here, we briefly explain how to apply this technique to $\SAS_{\CK}$.
We divide elements in the public key of $\SAS_{\CK}$ into shared elements $(\widetilde{H}, \widetilde{D}, \widetilde{U})$ and non-shared elements $(\widetilde{V}_{j})_{j \in [\ell]}$ and obtain our scheme $\SAS_{\OursCK}$.
The difference between $\SAS_{\CK}$ \cite{CK20} and $\SAS_{\OursCK}$ is the message space.
The message space of $\SAS_{\CK}$ is $\Z^{*}_{p}$ and $\SAS_{\OursCK}$ is extended to $(\Z^{*}_{p})^{\ell}$.
$\SAS_{\CK}$ with $\ell=1$ corresponds to $\SAS_{\OursCK}$.
We give the full description of our scheme $\SAS_{\OursCK}$ in Fig. \ref{SAS_OursCK_const}.

\begin{figure}[h]
\centering
\begin{tabular}{|l|}
\hline
$\nSASSetup(1^{\lambda}):$\\ 
~~~$\BGcal_{3}= (p, \G, \widetilde{\G}, \G_T, e) \leftarrow \BG(1^\lambda)$, $G \xleftarrow{\$} \G^*$, $\widetilde{G} \xleftarrow{\$} \widetilde{\G}^*$,\\
~~~$d, x_{1},  x_{2}  \xleftarrow{\$} \mathbb{Z}^{*}_{p}$, $X_{1} \leftarrow G^{x_{1}}$, $X_{2} \leftarrow G^{x_{2}}$, $\widetilde{H} \xleftarrow{\$} \widetilde{\G}$, $\widetilde{D} \leftarrow \widetilde{H}^{d}$, $\widetilde{U} \leftarrow \widetilde{H}^{x_{2} - dx_{1}}$.\\
~~~Return $\pp \leftarrow (\BGcal_{3}, G_{1}, G_{2}, X_{1}, X_{2}, \widetilde{H}, \widetilde{D}, \widetilde{U})$.\\

$\nSASKGen(\pp):$\\
~~~$y_{1, 1},  y_{1, 2} \xleftarrow{\$} (\mathbb{Z}^{*}_{p})^{2}, \textcolor{red}{(y_{j,1},  y_{j,2})_{j \in \{2, \dots \ell \}}  \xleftarrow{\$} (\mathbb{Z}^{*}_{p})^{2(\ell-1)}}$,\\
~~~$\widetilde{V}_{1} \leftarrow \widetilde{H}^{y_{1, 2}} ( \widetilde{D}^{y_{1, 1}})^{-1}$, $\textcolor{red}{(\widetilde{V}_{j} \leftarrow \widetilde{H}^{y_{j, 2}} ( \widetilde{D}^{y_{j, 1}})^{-1})_{j \in \{2, \dots \ell \}}}$.\\
~~~Return $(\pk, \sk) \leftarrow ((\widetilde{V}_{1}, \textcolor{red}{(\widetilde{V}_{j})_{j \in \{2, \dots \ell \}}}), (y_{1, 1},  y_{1, 2}, \textcolor{red}{(y_{j,1},  y_{j,2})_{j \in \{2, \dots \ell \}}}))$.\\

$\nSASKVerify(\pk = (\widetilde{V}_{1}, \textcolor{red}{(\widetilde{V}_{j})_{j \in \{2, \dots \ell \}}}), \sk = (y_{1, 1},  y_{1, 2}, \textcolor{red}{(y_{j,1},  y_{j,2})_{j \in \{2, \dots \ell \}}})):$\\
~~~If $\widetilde{V}_{1} = \widetilde{H}^{y_{1, 2}} (\widetilde{D}^{y_{1, 1}})^{-1} \land$ \textcolor{red}{$\widetilde{V}_{j} = \widetilde{H}^{y_{j, 2}} (\widetilde{D}^{y_{j, 1}})^{-1}$ for $j \in [\ell] \backslash \{1\}$}, return $1$.\\
~~~Otherwise return $0$.\\

$\nSASSign(\sk_{n}= (y_{n, 1, 1},  y_{n, 1, 2}, \textcolor{red}{(y_{n, j, 1},  y_{n, j, 2})_{j \in \{2, \dots \ell \}}}), L_{n-1},$\\
~~~~~~~~~$(m_{i} = (m_{i,1}, \textcolor{red}{(m_{i, j})_{j \in \{2, \dots, \ell\}}}))_{i \in [n-1]},$ $m_{n} =  (m_{n, 1}, \textcolor{red}{(m_{i, j})_{j \in \{2, \dots, \ell\}}}),$\\
~~~~~~~~~~~~~~~~~~~~~~~~~~~~~~~~~~~~~~~~~~~~~~~~~~~~~~~~~~~~~~~$\sigma_{n-1}= (A_{n-1}, B_{n-1}, C_{n-1})):$\\
~~~If $m_{n} = 0$, return $\bot$.\\
~~~If $L_{n-1} = \epsilon$ (i.e., $n-1=0$), $\sigma_{0} = (A_{0}, B_{0}, C_{0}) \leftarrow(G, X_{1}, X_{2})$.\\
~~~If $L_{n-1} \neq \epsilon$, \\
~~~~~If $\nSASSVerify(L_{n-1}, (m_{1}, \dots, m_{n-1}), \sigma_{n-1}) = 0$, return $\bot$.\\
~~~~~If there exists $(j, j')$ such that $j \neq j' \land \pk_{j} = \pk_{j'}$, return $\bot$. \\
~~~$r_{n} \xleftarrow{\$} \mathbb{Z}^{*}_{p}$, $A_{n} \leftarrow A^{r_{n}}_{n-1}$, $B_{n} \leftarrow (B_{n-1} A^{m_{n, 1}y_{n, 1, 1} \textcolor{red}{+ \sum_{j \in \{2, \dots, \ell\}} m_{n, j}y_{n, j,1}}}_{n-1})^{r_{n}}$,\\ ~~~$C_{n} \leftarrow (C_{n-1} A^{m_{n, 1}y_{n, 1, 2} \textcolor{red}{+ \sum_{j \in \{2, \dots, \ell\}} m_{n, j}y_{n, j,2}}}_{n-1})^{r_{n}}$.\\
~~~Return $\sigma_{n} \leftarrow (A_{n}, B_{n}, C_{n})$.\\

$\nSASSVerify(L_{n}=(\pk_{i} = (\widetilde{V}_{i, 1}, \textcolor{red}{(\widetilde{V}_{i, j})_{j \in \{2, \dots \ell \}}})_{i \in [n]},$\\
~~~~~~~~~~~~~~~~~~$(m_{i} = (m_{i,1}, \textcolor{red}{(m_{i, j})_{j \in \{2, \dots, \ell\}}}))_{i \in [n-1]}, \sigma=(A, B, C)):$\\
~~~If $m_{i, 1} \neq 0,  \textcolor{red}{m_{i, j} \neq 0}$ for $i \in [n], \textcolor{red} {j \in \{2, \dots, \ell\}}$  $\land A \neq 1_{\G}$\\
~~~~~~$\land e(A, \widetilde{U} \prod_{i \in [n]} (\widetilde{V}_{i, 1}^{m_{i, 1}} \textcolor{red}{\prod_{j \in \{2, \dots, \ell \} }\widetilde{V}^{m_{i, j}}_{i, j}}) \cdot e(B, \widetilde{D}) = e(C, \widetilde{H})$, return $1$.\\
~~~Otherwise return $0$.\\
\hline
\end{tabular}
\caption{\small The sequential aggregate signature scheme $\SAS_{\CK}$ \cite{CK20} and $\SAS_{\OursCK}$. 
The differences between $\SAS_{\CK}$ and $\SAS_{\OursCK}$ are highlighted in \textcolor{red}{red}. $\SAS_{\CK}$ corresponds to the part written in black.
$\SAS_{\OursCK}$ is obtained by adding \textcolor{red}{the elements and processes colored by red} to $\SAS_{\CK}$.}
\label{SAS_OursCK_const}
\end{figure}

\begin{theorem}\label{Th_ASAOursCK_EUF_from_DSOursCK_EUF}
If $\DS_{\CK}$ satisfies the EUF-CMA security and the DBP assumption on $\G$ holds, then $\SAS_{\OursCK}$ satisfies the EUF-CMA security in the certified key model.
\end{theorem}

Below, we prove Theorem \ref{Th_ASAOursCK_EUF_from_DSOursCK_EUF} by extending the security proof of \cite{CK20}.
The difference between our security proof and their security proof is the precise discussion of a signature distribution in a signing query.
If an adversary $\A$ of $\SAS_{\OursCK}$ makes a signature query that satisfies the specific condition, the simulated signatures have a different distribution from the original EUF-CMA game.
In our security proof, we introduce the event $\Bad$ that $\A$ makes a query with this condition.
Then, we prove that $\Bad$ occurs with negligible probability in Lemma \ref{Lem_Bad_Does_not_occurs}.
This lemma fills the logical gap in their security proof.

\begin{proof}
Let $\A$ be a PPT adversary for the EUF-CMA security of $\SAS_{\OursCK}$ in the certified key model.
First, we give a reduction algorithm $\B$ for the EUF-CMA security of $\DS_{\OursCK}$ as follows.

\begin{itemize}
\item {\bf Initial setup:}
$\B$ takes as an input an instance $(\pp, \pk^{*}) = ((\BGcal_{3}, G, \widetilde{G}), (\widetilde{H}, \allowbreak \widetilde{D}, \widetilde{U}, (\widetilde{V}_{j})_{j \in [\ell]} \allowbreak \widetilde{W})$ of the EUF-CMA security game of $\DS_{\CK}$.
$\B$ initialize $S^{\Cert} \leftarrow \{\}$, $S^{\Sign} \leftarrow \{\}$, makes signing query on a message $(0,\dots 0)$ to the signing oracle for $\DS_{\CK}$, and obtains a signature $\sigma' = (A', B', C')$.
$\B$ sets $G' \leftarrow A'$, $X_{1} \leftarrow B'$, $X_{2} \leftarrow C'$.
Then, $\B$ sends $(\pp', \pk') = ((\BGcal_{3}, G', \widetilde{G}, X_{1}, X_{2}, \allowbreak \widetilde{H}, \allowbreak \widetilde{D}, \widetilde{U}), ((\widetilde{V}_{j})_{j \in [\ell]}))$ to $\A$.

\item {\bf Key registration query:}
For a key registration query on $(\pk, \sk)$, $\B$ checks $\nOMSKVerify(\pk , \sk) \allowbreak = 1$.
If this condition holds, $\B$ updates $S^{\Cert} \leftarrow S^{\Cert} \cup \{(\pk , \sk)\}$ and returns $\accept$.
Otherwise, $\B$ returns $\reject$.

\item {\bf Signing query:}
For a signing query on $(L_{n-1} = (\pk_{1}, \dots, \pk_{n-1}), (m_{i} = (m_{i, j})_{j \in [\ell]})_{i \in [n-1]}$, $m_{n} =  (m_{n, j})_{j \in [\ell]}), \sigma_{n-1}=(A_{n-1}, B_{n-1}, C_{n-1})$, if $L_{n-1}  \neq  \epsilon$ (i.e., $n-1 \geq 1$), $\B$ checks that the following conditions:
\begin{enumerate}
\item $\pk^{*} \notin L_{n-1} ;$
\item $\pk_{i} \in Q^{\Cert}$ for $i \in [n-1] ;$
\item $\nOMSSVerify(L_{n-1}, (m_{i})_{i \in [n-1]}, \sigma_{n-1}) = 1$.
\end{enumerate}
If at least one of the above conditions does not hold, $\B$ returns $\bot$.
Then, for $i \in [n-1]$, $\B$ retrieves $(\pk_{i}, \sk_{i} = ((y_{i, j,1},  y_{i, j,2})_{j \in [\ell]})$ from $Q^{\Cert}$.
$\B$ makes a signing query on a message $m_{n} =  (m_{n, j})_{j \in [\ell]}$ to the signing oracle for $\DS_{\CK}$, and obtains a signature $(A', B', C')$. 
Then $\B$ computes $A_{n} \leftarrow A'$, $B_{n} \leftarrow B' (A')^{\sum_{i \in [n-1]} \sum_{j \in [\ell]}m_{i, j}y_{i, j, 1}}$, $C_{n} \leftarrow C' (A')^{\sum_{i \in [n-1]} \sum_{j \in [\ell]}m_{i, j}y_{i, j, 2}}$.
Then $\B$ updates $S^{\Sign} \leftarrow S^{\Sign} \cup \{m_{n}\}$.
$\B$ returns $\sigma_{n} = (A_{n}, B_{n}, C_{n})$ to $\A$.

\item {\bf End of the game:}
After receiving the forgery $(L^{*}_{n^{*}}=(\pk^{*}_{1}, \dots, \pk^{*}_{n^{*}}), (m^{*}_{i} = (m^{*}_{i, j})_{j \in [\ell]})_{i \in [n^{*}]}, \allowbreak \sigma^{*}_{n^{*}} = (A^{*}_{n^{*}}, B^{*}_{n^{*}}, C^{*}_{n^{*}}))$ from $\A$, 
$\B$ checks the following conditions hold.
\begin{enumerate}
\item $\nSASSVerify(L^{*}_{n^{*}}, (m^{*}_{i})_{i \in [n^{*}]}, \sigma^{*}_{n^{*}}) = 1;$
\item There exists $i^{*} \in [n^{*}]$ such that $\pk^{*}_{i^{*}} = \pk^{*} \land m_{i^{*}}^{*} \notin S^{\Sign};$
\item $\forall j \in [n^{*}]$ such that $\pk_{j} \neq \pk^{*}$, $\pk_{j} \in Q^{\Cert}.$
\end{enumerate}
For $i \in [n^{*}] \backslash \{i^{*}\}$, $\B$ retrieves $(\pk^{*}_{i}, \sk^{*}_{i} = ((y_{i, j,1},  y_{i, j,2})_{j \in [\ell]})$ from $Q^{\Cert}$.
Then, $\B$ computes $\hat{A} \leftarrow A^{*}_{n^{*}}$, $\hat{B} \leftarrow B^{*}_{n^{*}}((A^{*}_{n^{*}})^{\sum_{i \in [n^{*}] \backslash \{i^{*}\}} \sum_{j \in [\ell]}m^{*}_{i, j}y_{i, j, 1}})^{-1}$, $C^{*} \leftarrow \hat{C}_{n^{*}}((A^{*}_{n^{*}})^{\sum_{i \in [n^{*}] \backslash \{i^{*}\}} \sum_{j \in [\ell]}m^{*}_{i, j}y_{i, j, 2}})^{-1}$.
Finally, $\B$ outputs $(\hat{m}, \hat{\sigma}) = (m_{i^{*}}^{*}, \allowbreak (\hat{A}, \hat{B}, \hat{C}))$ as a forgery of $\DS_{\CK}$.
\end{itemize}

Next, we discuss the EUF-CMA security game simulation by $\B$.
Let $\Bad$ be the event that $\A$ submits a valid signature $(L_{n-1}, \allowbreak (m_{i})_{i \in [n-1]}$, $m_{n}, \sigma_{n-1})$ such that $B_{n-1} \neq A^{x_{1} + \sum_{i \in [n-1]} \sum_{j \in [\ell]} m_{i, j}y_{i, j, 1} }_{n-1}$ for some signing query.
We confirm that if the event $\Bad$ does not occur, $\B$ simulates the EUF-CMA game of $\SAS_{\OursCK}$.
\begin{itemize}
\item {\bf Initial setup:}
In the initial setup, the difference between the EUF-CMA game of $\SAS_{\OursCK}$ and the simulation by $\B$ is a generation of $(G', X_{1}, X_{2})$. 
In the original EUF-CMA game of $\SAS_{\OursCK}$, these elements are generated as $G' \xleftarrow{\$} {\G}^{*}$, $x_{1}, x_{2} \xleftarrow{\$} \Z_{p}$, $X_{1}, \leftarrow G^{x_{1}}$, and $X_{2}, \leftarrow G^{x_{2}}$.
In the simulation of $\B$, these elements are generated as $G' \leftarrow A'$, $X_{1}, \leftarrow B'$, and $X_{2}, \leftarrow C'$ where $\sigma' = (A', B', C')$ is the signature of $\DS_{\CK}$ on a message $(0, \dots, 0)$.
Since $\sigma' = (A', B', C')$ is generated as $r \xleftarrow{\$} \mathbb{Z}^{*}_{p}$, $A' \leftarrow G^{r}$, $B' \leftarrow (A')^{x_{1}}$, $C' \leftarrow (A')^{x_{2}}$, the distributions of $(G', X_{1}, X_{2})$ between the original game and the simulated game by $\B$ are identical.
Thus, $\B$ simulates $(\pp', \pk')$.

\item {\bf Key registration query:} Clearly, $\B$ simulates a key registration.

\item {\bf Signing query:}
For a signing query on $(L_{n-1} = (\pk_{1}, \dots, \pk_{n-1}), (m_{i} = (m_{i, j})_{j \in [\ell]})_{i \in [n-1]}$, $m_{n} =  (m_{n, j})_{j \in [\ell]}, \sigma_{n-1}= (A_{n-1}, B_{n-1}, C_{n-1}))$, if $L_{n-1}  \neq  \epsilon$, the difference between the the EUF-CMA game of $\SAS_{\OursCK}$ and the simulation by $\B$ is generation of a signature $\sigma_{n} = (A_{n}, B_{n}, C_{n})$. 

In the simulated game by $\B$, $\sigma_{n} = (A_{n}, B_{n}, C_{n})$ is generated as follows.
$\B$ makes a signing query on a message $m_{n} =  (m_{n, j})_{j \in [\ell]}$ to the signing oracle for $\DS_{\CK}$, and obtains a signature $(A', B', C') = (A', (A')^{x_{1} + \sum_{j \in [\ell]}m_{i, j}y_{i, j, 1}}, \allowbreak (A')^{x_{2} + \sum_{j \in [\ell]}m_{i, j}y_{i, j, 1}})$ where $((y_{n j,1},  y_{n, j,2})_{j \in [\ell]}) = \sk^{*}$.
Then $\B$ computes $B_{n} \leftarrow B' (A')^{\sum_{i \in [n-1]} \sum_{j \in [\ell]}m_{i, j}y_{i, j, 1}}$, $C_{n} \leftarrow C' (A')^{\sum_{i \in [n-1]} \sum_{j \in [\ell]}m_{i, j}y_{i, j, 2}}$.
The output is $\sigma_{n} = (A_{n}, B_{n}, C_{n}) =  (A', (A')^{x_{1} + \sum_{i \in [n]} \sum_{j \in [\ell]} m_{i, j}y_{i, j, 1} }, \allowbreak (A')^{x_{2} + \sum_{i \in [n]} \sum_{j \in [\ell]} m_{i, j}y_{i, j, 2} })$.

In the original EUF-CMA game of $\SAS_{\OursCK}$, $\sigma_{n} = (A_{n}, B_{n}, C_{n})$ is generated as $r_{n} \xleftarrow{\$} \mathbb{Z}^{*}_{p}$, $A_{n} \leftarrow A^{r_{n}}_{n-1}$, $B_{n} \leftarrow (B_{n-1} A^{\sum_{j \in [\ell]} m_{n, j}y_{n, j,1}}_{n-1})^{r_{n}}$, $C_{n} \leftarrow (C_{n-1} A^{\sum_{j \in [\ell]} m_{n, j}y_{n, j,2}}_{n-1})^{r_{n}}$ where $((y_{n j,1},  y_{n, j,2})_{j \in [\ell]}) = \sk^{*}$.
In the case of $\lnot \Bad$, $B_{n-1} = A^{x_{1} + \sum_{i \in [n-1]} \sum_{j \in [\ell]} m_{i, j}y_{i, j, 1} }_{n-1}$, $C_{n-1} = A^{x_{2} + \sum_{i \in [n-1]} \sum_{j \in [\ell]} m_{i, j}y_{i, j, 2} }_{n-1}$ holds.
The output signature in the original EUF-CMA game is $\sigma_{n} = (A_{n}, B_{n}, \allowbreak C_{n}) =  (A_{n}, \allowbreak A^{x_{1} + \sum_{i \in [n]} \sum_{j \in [\ell]} m_{i, j}y_{i, j, 1} }_{n-1}, \allowbreak  A^{x_{2} + \sum_{i \in [n]} \sum_{j \in [\ell]} m_{i, j}y_{i, j, 2} }_{n-1})$.
If $\Bad$ does not occur, the distributions of a signature $\sigma_{n}$ simulated by $\B$ and output by signing oracle in the original EUF-CMA game are identical.
\end{itemize}
Thus, if $\Bad$ does not occurs, $\B$ simulates the EUF-CMA game of $\SAS_{\OursCK}$.

Next, we consider the signing simulation in the event $\Bad$.
The signing simulation of our proof for $\SAS_{\Ours}$ with $\ell=1$ and the proof by Chatterjee and Kabaleeshwaran \cite{CK20} is the same.
They claimed that the distribution of simulated signatures is ``It is easy to see that the signature generated above is properly distributed.''. 
We have already confirmed their claim is correct in the case where $\Bad$ does not occur.
If $\Bad$ occurs, their claim is non-trivial since the distribution of simulated signatures by $\B$ is different from the original game.
However, they did not discuss the distribution of signatures in this case.
To fill this gap, we prove that $\Bad$ occurs with negligible probability in Lemma~\ref{Lem_Bad_Does_not_occurs}.

\begin{lemma}\label{Lem_Bad_Does_not_occurs}
If the double pairing (DBP) assumption on $\G$ holds (See Appendix \ref{Appen_BG} for the DBP assumption), the event $\Bad$ occurs with negligible probability in~$\lambda$.
\end{lemma}

\begin{proof}
We give the reduction $\R$ for the DBP problem as follows.

\begin{itemize}
\item {\bf Initial setup:}
$\R$ takes as an input an instance $(\BGcal_{3}, G, \widetilde{G}, \widetilde{H}, \widetilde{D})$ of the DBP problem.
$\R$ initialize $S^{\Cert} \leftarrow \{\}$, $S^{\Sign} \leftarrow \{\}$, chooses $x_{1}, x_{2} \xleftarrow{\$} \Z_{p}$, computes  $X_{1}, \leftarrow G^{x_{1}}$, $X_{2}, \leftarrow G^{x_{2}}$, $\widetilde{U} \leftarrow \widetilde{H}^{x_{2}} \widetilde{D}^{-x_{1}}$ and sets $\pp' \leftarrow (\BGcal_{3}, G', \widetilde{G}, X_{1}, X_{2}, \widetilde{H}, \allowbreak \widetilde{D}, \widetilde{U})$.
Then, $\R$ runs $\pk \leftarrow \nSASKGen(\pp)$ and sends $(\pp, \pk)$ to $\A$.

\item {{\bf Key registration query:} Same as in the original EUF-CMA game.}

\item {{\bf Signing query:} Same as in the original EUF-CMA game except for the following procedure.}
For a signing query on $(L_{n-1} = (\pk_{1}, \dots, \pk_{n-1}), (m_{i} = (m_{i, j})_{j \in [\ell]})_{i \in [n-1]}$, $m_{n} =  (m_{n, j})_{j \in [\ell]}), \sigma_{n-1}=(A_{n-1}, B_{n-1}, C_{n-1})I$, if a signature $ \sigma_{n-1}$ is valid, $\R$ additionally checks 
the pairing equation $e(A_{n-1}, \allowbreak X_{1}\prod_{i \in [n-1]} \prod_{j \in [\ell]} (\widetilde{G}^{y_{i, j, 1}})^{m_{i, j}}) = e(B_{n-1}, \widetilde{G})$.
This check is possible in the certified key model since $\R$ knows $\sk_{i} = ((y_{i, j, 1}, \allowbreak y_{i, j, 2})_{j \in [\ell]})$ for all $i \in [\ell]$.

If this pairing equation holds (i.e., $\Bad$ occurs), $\R$ extracts a solution of the DBP problem as follows.
$\R$ computes $E \leftarrow C_{n-1}A^{x_{2} + \sum_{i \in [n-1]} \sum_{j \in [\ell]} m_{i, j} y_{i, j, 2}}_{n-1}$, $F \leftarrow B_{n-1}A^{- x_{1} - \sum_{i \in [n-1]} \sum_{j \in [\ell]} m_{i, j} y_{i, j, 1}}_{n-1}$.
Then, $\R$ outputs $(E, F)$ as a solution to the DBP problem.
\end{itemize}

We confirm that $\Bad$ occurs, $\R$ outputs a solution to the DBP problem.
If $\Bad$ occurs, $\A$ submits a valid signature $(L_{n-1}, \allowbreak (m_{i})_{i \in [n-1]}$, $m_{n}, \sigma_{n-1})$ such that $B_{n-1} \neq  A^{x_{1} + \sum_{i \in [n-1]} \sum_{j \in [\ell]} m_{i, j}y_{i, j, 1} }_{n-1}$ for some signing query.
If the signature is valid, $\nSASSVerify(L_{n-1}, (m_{1}, \dots, m_{n-1}), \sigma_{n-1}) = 1$ holds.
By this fact, we see that $A_{n-1} \neq 1_{\G}$ and $e(A_{n-1}, \widetilde{U} \prod_{i \in [n-1]}  \prod_{j \in [\ell] }\widetilde{V}^{m_{i, j}}_{i, j}) \cdot e(B_{n-1}, \widetilde{D}) e(C_{n-1}, \widetilde{H}) ^{-1} \allowbreak = 1_{\G_{T}}$ hold.
In the certified key model, we can assure that $\widetilde{V}_{i, j} = \widetilde{H}^{y_{i, j, 2}} (\widetilde{D}^{y_{i, j, 1}})^{-1}$ holds for $i \in [n-1]$ and $j \in [\ell]$.
Then we have
\begin{equation*}
\begin{split}
&e(A_{n-1}, \widetilde{U} \prod_{i \in [n-1]}  \prod_{j \in [\ell] }\widetilde{V}^{m_{i, j}}_{i, j}) \\
&= e(A_{n-1}, \widetilde{H}^{x_{2}} \widetilde{D}^{-x_{1}}\prod_{i \in [n-1]}  \prod_{j \in [\ell] }(\widetilde{H}^{m_{i, j} y_{i, j, 2}} (\widetilde{D}^{-m_{i, j} y_{i, j, 1}})))\\
&= e(A^{x_{2} + \sum_{i \in [n-1]} \sum_{j \in [\ell]} m_{i, j} y_{i, j, 2}}_{n-1}, \widetilde{H}) \cdot e(A^{- x_{1} - \sum_{i \in [n-1]} \sum_{j \in [\ell]} m_{i, j} y_{i, j, 1}}_{n-1}, \widetilde{D}).
\end{split}
\end{equation*}
From this fact, we see that 
\begin{equation*}
\begin{split}
&e(A_{n-1}, \widetilde{U} \prod_{i \in [n-1]}  \prod_{j \in [\ell] }\widetilde{V}^{m_{i, j}}_{i, j}) \cdot e(B_{n-1}, \widetilde{D}) e(C_{n-1}, \widetilde{H}) ^{-1}\\
&=e(C_{n-1}A^{x_{2} + \sum_{i \in [n-1]} \sum_{j \in [\ell]} m_{i, j} y_{i, j, 2}}_{n-1}, \widetilde{H}) \\
&~~~~~~\cdot e(B_{n-1}A^{- x_{1} - \sum_{i \in [n-1]} \sum_{j \in [\ell]} m_{i, j} y_{i, j, 1}}_{n-1}, \widetilde{D})\\
&= e(E, \widetilde{H}) \cdot e(F, \widetilde{D}) = 1_{\G_{T}}
\end{split}
\end{equation*}
holds. 
Moreover, in the case where $\Bad$ occurs, $B_{n-1} \neq  A^{x_{1} + \sum_{i \in [n-1]} \sum_{j \in [\ell]} m_{i, j}y_{i, j, 1} }_{n-1}$ holds.
This implies that $(E, F) \neq (1_{\G}, 1_{\G})$ holds. 
We see that $(E, F)$ is a solution to the DBP problem.
Thus, If $\Bad$ occurs, $\R$ solves the DBP problem.
We conclude Lemma~\ref{Lem_Bad_Does_not_occurs}.
 \qed
\end{proof}

Third, we confirm that $\B$ extracts a valid forgery for $\DS_{\CK}$ when the event $\Bad$ does not occur and $\A$ outputs a valid forgery for $\SAS_{\CK}$.
Let $(L^{*}_{n^{*}}=(\pk^{*}_{1}, \dots, \pk^{*}_{n^{*}}), (m^{*}_{i} = (m^{*}_{i, j})_{j \in [\ell]})_{i \in [n^{*}]}, \allowbreak \sigma^{*}_{n^{*}} = (A^{*}_{n^{*}}, B^{*}_{n^{*}}, C^{*}_{n^{*}}))$ be a valid forgery output by $\A$.
If $\nSASSVerify(L^{*}_{n^{*}}, (m_{1}^{*}, \dots, m^{*}_{n^{*}}), \sigma^{*}_{n^{*}}) = 1$ holds, the following equations hold:
\begin{equation*}
\begin{split}
&A^{*} \neq 1_{\G}; (A^{*}_{n^{*}})^{x_{2} - dx_{1} + \sum_{i \in [n^{*}]} \sum_{j \in [\ell]} m^{*}_{i, j}(y_{i, j, 2} - dy_{i, j, 1})} = C^{*}_{n^{*}}(B^{*}_{n^{*}})^{-d};\\
&B^{*}_{n^{*}} = (A^{*}_{n^{*}})^{x_{1} + \sum_{i \in [n]} \sum_{j \in [\ell]} m^{*}_{i, j}y_{i, j, 1}}.\\
\end{split}
\end{equation*}
This fact implies that the following equations hold:
\begin{equation*}
\begin{split}
&A^{*} \neq 1_{\G}; B^{*}_{n^{*}} = (A^{*}_{n^{*}})^{x_{1} + \sum_{i \in [n]} \sum_{j \in [\ell]} m^{*}_{i, j}y_{i, j, 1}};\\
&C^{*}_{n^{*}} =(A^{*}_{n^{*}})^{x_{2} + \sum_{i \in [n]} \sum_{j \in [\ell]} m^{*}_{i, j}y_{i, j, 2}}.
\end{split}
\end{equation*}

$\hat{\sigma} = (\hat{A}, \hat{B}, \hat{C})$ is computed as 
\begin{equation*}
\begin{split}
&\hat{A} \leftarrow A^{*}_{n^{*}}; \hat{B} \leftarrow B^{*}_{n^{*}}((A^{*}_{n^{*}})^{\sum_{i \in [n^{*}] \backslash \{i^{*}\}} \sum_{j \in [\ell]}m^{*}_{i, j}y_{i, j, 1}})^{-1}\\
&\hat{C} \leftarrow \hat{C}_{n^{*}}((A^{*}_{n^{*}})^{\sum_{i \in [n^{*}] \backslash \{i^{*}\}} \sum_{j \in [\ell]}m^{*}_{i, j}y_{i, j, 2}})^{-1}.
\end{split}
\end{equation*}

From these facts, we see that 
\begin{equation*}
\begin{split}
&\hat{A} = A^{*}_{n^{*}}; \hat{B} = (A^{*}_{n^{*}})^{x_{1} + \sum_{j \in [\ell]} m^{*}_{i^{*}, j}y_{i^{*}, j, 1}}; \hat{C} = (A^{*}_{n^{*}})^{x_{2} + \sum_{j \in [\ell]} m^{*}_{i^{*}, j}y_{i^{*}, j, 2}}.
\end{split}
\end{equation*}
holds where $\sk^{*} = ((y_{i^{*}, j, 1}, y_{i^{*}, j, 2})_{j \in [\ell]})$.
Since $m^{*}$ is not queried to signing, we see that $(\hat{m},\hat{\sigma}) = (m^{*}_{i}, \hat{A}, \hat{B}, \hat{C})$ is a valid forgery for $\DS_{\CK}$.

Finally, we bound the advantage $\Adv^{\DBP_{\G}}_{\BG, \R}(\lambda)$.
From Lemma \ref{Lem_Bad_Does_not_occurs}, we have
\begin{equation*}
\begin{split}
\Adv^{\EUFCMACK}_{\SAS, \A}(\lambda) \allowbreak  &= \Pr[\lnot\Bad] \Pr[\sfGame^{\EUFCMACK}_{\SAS, \A}(1^{\lambda}) \Rightarrow 1 | \lnot\Bad]\\
&~~~~~~ + \Pr[\Bad] \Pr[\sfGame^{\EUFCMACK}_{\SAS, \A}(1^{\lambda}) \Rightarrow 1 | \Bad]\\
& = \Adv^{\DBP_{\G}}_{\BG, \R}(\lambda) + \negl(\lambda).
\end{split}
\end{equation*}

Thus, we conclude Theorem \ref{Th_ASAOursCK_EUF_from_DSOursCK_EUF}.
\qed
\end{proof}
From Lemma \ref{Th_DSCK_EUF_from_SXDH},  Theorem \ref{Th_ASAOursCK_EUF_from_DSOursCK_EUF}, and Lemma \ref{Lem_from_SXDH_to_DBP}, we obtain the following fact.

\begin{corollary}\label{Coro_SXDH_SAS_Ours}
If the SXDH assumption holds, then $\SAS_{\OursCK}$ satisfies EUF-CMA security in the certified key model.
\end{corollary}

\section{Ordered Multi-Signatures (OMS)}\label{Sec_OMS}
In this section, we introduce a definition of an ordered multi-signature scheme with public-key aggregation and its security notion.
Then, we propose our ordered multi-signature scheme $\OMS_{\OursCK}$.

\subsection{Ordered Multi-Signature Scheme $\OMS$}

\begin{definition}[Ordered Multi-Signature Scheme]
Let $n_{max}(\lambda) = \poly(\lambda)$ be a polynomial.\footnote{$n_{max}(\lambda)$ represents the maximum number of signers that participate in signing for each signature.}
An ordered multi-signature scheme $\OMS$ consists of the following tuple of algorithms $(\nOMSSetup, \nOMSKGen, \allowbreak \nOMSKVerify, \allowbreak \nOMSKAgg, \nOMSSign, \nOMSSVerify)$.
\begin{itemize}
\item $\nOMSSetup (1^{\lambda}):$ A setup algorithm takes as an input a security parameter $1^{\lambda}$. It returns the public parameter $\pp$.
In this work, we assume that $\pp$ defines a message space and represents this space by $\mathcal{M}$.
We omit a public parameter $\pp$ in the input of all algorithms except for $\nOMSKGen$.

\item $\nOMSKGen (\pp):$ A key generation algorithm takes as an input a public parameter $\pp$. It returns a public key $\pk$ and a secret key $\sk$.

\item $\nOMSKVerify(\pk, \sk):$ A key verification algorithm takes as an input a public key $\pk$ and a secret key $\sk$. 
It returns a bit $b \in  \{0, 1\}$.

\item $\nOMSKAgg(L_{n}=(\pk_{1}, \dots, \pk_{n})):$ A key aggregation algorithm takes as an input a list of public key $L_{n}=(\pk_{1}, \dots, \pk_{n})$. 
It returns an aggregated public key $\apk$ or $\bot$.

\item $\nOMSSign(\sk_{n}, L_{n-1}=(\pk_{1}, \dots, \pk_{n-1}), m, \sigma_{n-1}):$ A signing algorithm takes as an input a secret key $\sk_{n}$, a list of public keys $L_{n-1}=(\pk_{1}, \dots, \pk_{n-1})$, a message $m$, and a signature $\sigma_{n-1}$. 
It returns an updated signature $\sigma_{n}$ or $\bot$.

\item $\nOMSSVerify (\apk, m, \sigma):$ A signature verification algorithm takes as an input an aggregated public key $\apk$, a message $m$, and a signature $\sigma$.
It returns a bit $b \in  \{0, 1\}$.
\end{itemize}
\end{definition}
\paragraph{\bf Correctness.}
$\OMS$ satisfies correctness if $\forall \lambda \in \N$, $\pp \leftarrow \nOMSSetup (1^{\lambda})$, $\forall n  \leq n_{max}(\lambda)$,  $\forall m \in \mathcal{M}$, $(\pk_{i}, \sk_{i}) \leftarrow \nOMSKGen(\pp)$ for $i \in [n]$, $L_{0} = \epsilon$, $\sigma_{0} = \epsilon$, $L_{i}=(\pk_{1}, \dots, \pk_{i})$ for $i \in [n]$, and $\sigma_{i} \leftarrow \nOMSSign(\sk_{i}, L_{i-1}, m, \sigma_{i-1})$ for $i \in [n]$, 
\begin{enumerate}
\item For $i \in [n]$, $\nOMSKVerify(\pk_{i}, \sk_{i}) = 1$
\item For $i \in [n]$, if elements in $L_{i}$ are distinct, $\nOMSSVerify(\nOMSKAgg(L_{i}), m, \sigma_{i}) = 1$
\end{enumerate}
holds.

We review the existentially unforgeable under chosen message attacks (EUF-CMA) security in the certified key model \cite{BGOY07}.
This security ensures that it is hard for any PPT adversary not only to forge a signature for a new message but also to forge a signature that is swapped to the order of honest signers. 
Similar to Definition~\ref{Def_SAS_EUF_CK}, we define the security in the certified model.

\begin{definition}[EUF-CMA Security in CK model]
Let $\OMS$ be an ordered multi-signature scheme and $\A$ be a PPT adversary.
The existentially unforgeable under chosen message attacks (EUF-CMA) security in the certified key model is defined by the following game $\sfGame^{\EUFCMACK}_{\OMS, \A}(1^{\lambda})$ between the challenger $\C$ and an adversary $\A$.

\begin{itemize}
\item {\bf Initial setup:}
$\C$ initializes sets $S^{\Cert} \leftarrow \{\}$, $S^{\Sign} \leftarrow \{\}$, runs $\pp \leftarrow \nOMSSetup (1^{\lambda})$, $(\pk^{*}, \sk^{*}) \leftarrow \nOMSKGen(\pp)$, and sends $(\pp, \pk^{*})$ to $\A$.
\item $\A$ makes queries for the following oracles $\mathcal{O}^{\Cert}$ and $\mathcal{O}^{\Sign}$ polynomially many times.
\begin{itemize}
\item {\bf Key registration query:} 
For a key registration query on $(\pk, \sk)$, $\mathcal{O}^{\Cert}$ checks the validity of $(\pk, \sk)$.
If $\nOMSKVerify(\pk, \sk) = 1$, $\mathcal{O}^{\Cert}$ updates $S^{\Cert} \leftarrow S^{\Cert} \cup \{\pk\}$ and returns $\accept$.
Otherwise, $\mathcal{O}^{\Cert}$ returns $\reject$.

\item {\bf Signing query:}
For a signing query on $(L_{n-1} = (\pk_{1}, \dots, \pk_{n-1}), m, \sigma_{n-1})$, $\mathcal{O}^{\Sign}$ proceeds as follows. 
If $L_{n-1}  \neq  \epsilon$ (i.e., $n-1 \geq 1$), $\mathcal{O}^{\Sign}$ checks that the following conditions:
\begin{enumerate}
\item $\pk^{*} \notin L_{n-1} ;$
\item $\pk_{i} \in Q^{\Cert}$ for $i \in [n-1] ;$
\item $\nOMSSVerify(\nOMSKAgg(L_{n-1}), m, \sigma_{n-1}) = 1$.
\end{enumerate}
If at least one of the above conditions does not hold, $\mathcal{O}^{\Sign}$ returns $\bot$.
Then $\mathcal{O}^{\Sign}$ runs $\sigma_{n} \leftarrow \nOMSSign(\sk^{*}, L_{n-1}, m, \sigma_{n-1})$ and updates $S^{\Sign} \leftarrow S^{\Sign} \cup \{(m, n)\}$, and returns $\sigma_{n}$ to $\A$.
\end{itemize}
\item {\bf End of the game:}
$\A$ finally outputs a forgery $(L^{*}_{n^{*}}=(\pk^{*}_{1}, \dots, \pk^{*}_{n^{*}}), m^{*}, \sigma^{*}_{n^{*}})$ to $\C$.\\
If the following conditions hold, return $1$.
\begin{enumerate}
\item $n^{*} \leq n_{max} \land \nOMSSVerify(\nOMSKAgg(L^{*}_{n^{*}}), m^{*}, \sigma^{*}_{n^{*}}) = 1;$
\item There exists $i^{*} \in [n^{*}]$ such that $\pk^{*}_{i^{*}} = \pk^{*} \land (m^{*}, i^{*}) \notin S^{\Sign};$
\item $\forall j \in [n^{*}]$ such that $\pk_{j} \neq \pk^{*}$, $\pk_{j} \in Q^{\Cert}.$
\end{enumerate}
\end{itemize}

The advantage of an adversary $\A$ for the game is defined by $\Adv^{\EUFCMACK}_{\OMS, \A}(\lambda) \allowbreak:=$ $\allowbreak \Pr[\sfGame^{\EUFCMACK}_{\OMS, \A}(1^{\lambda}) \Rightarrow 1]$.
$\OMS$ satisfies the EUF-CMA security in the CK model if for any PPT adversary $\A$, $\Adv^{\EUFCMACK}_{\OMS, \A}(\lambda)$ is negligible in $\lambda$.
\end{definition}

\subsection{Derivation of $\OMS_{\OursCK}$ from $\SAS_{\OursCK}$ with $\ell=2$}\label{Subsec_OMS_Over_Ours_View}
We explain how to obtain $\OMS_{\OursCK}$ from $\SAS_{\OursCK}$ with $\ell = 2$.
First, we apply the transformation by Boldyreva et al. \cite{BGOY07} to $\SAS_{\OursCK}$ $\ell = 2$.
This transformation forces the signer $i$ to sign the message $(m_{i, 1}, m_{i, 2}) = (m, i)$ where $m$ is a common message for all signers.
The verification of a signature $\sigma_{n}  = (A, B, C)$ is done by checking a pairing equation $e(A, \widetilde{U} \prod_{i \in [n]} (\widetilde{V}_{i, 1})^{m} \widetilde{V}^{i}_{i, 2}) \cdot e(B, \widetilde{D}) = e(C, \widetilde{H})$ where $\pk_{i} = (\widetilde{V}_{i, 1}, \widetilde{V}_{i, 2})$.
Then, we modify this derived scheme to support public-key aggregation.
We compress the public key list $L_{n}$ into an aggregated public key $\apk = (\widetilde{\K}_{1}, \widetilde{\K}_{2}) = (\prod_{i \in [n]}\widetilde{V}_{i, 1}, \prod_{i \in [n]}\widetilde{V}^{i}_{i, 2})$.
We change the pairing equation in the verification as  $e(A, \widetilde{U} \prod_{i \in [n]} \widetilde{\K}_{1}^{m} \widetilde{\K}_{2}) \cdot e(B, \widetilde{D}) = e(C, \widetilde{H})$.
Thus, we obtain our scheme $\OMS_{\OursCK}$.

We recall that the ordered multi-signature scheme with the public key aggregation property is not obtained by applying Boldyreva et al. transformation to $\SAS_{\OursCK}$ of $\ell=1$ in Section \ref{Subsec_Moti_Pro}.
The important difference derived ordered multi-signature schemes between $\SAS_{\OursCK}$ of $\ell=1$ and $\ell=2$ is the components $\widetilde{U} \prod_{i \in [n]}\widetilde{V}_{i}^{m||i}$ and $\widetilde{U} \prod_{i \in [n]} (\widetilde{V}_{i, 1})^{m} \widetilde{V}^{i}_{i, 2}$ of the pairing equations that appeared in signature verification. 
Thanks to the algebraic structure $\prod_{i \in [n]} (\widetilde{V}_{i, 1})^{m} \widetilde{V}^{i}_{i, 2}$ in the case of $\ell=2$, we obtain the public key aggregation property from $\SAS_{\OursCK}$ with $\ell = 2$.

\subsection{Schemes $\OMS_{\OursCK}$}\label{Subsec_OMS_Ours_CK}
Let $n_{max}(\lambda) = \poly(\lambda)$ such that $n_{max}(\lambda) < 2^{\lambda-1}-1$.
Our ordered multi-signature scheme $\OMS_{\OursCK}$ with public-key aggregation is given in Fig. \ref{OMS_OursCK_const}.

\begin{figure}[h]
\centering
\begin{tabular}{|l|}
\hline
$\nOMSSetup(1^{\lambda}):$\\ 
~~~$\BGcal_{3}= (p, \G, \widetilde{\G}, \G_T, e) \leftarrow \BG(1^\lambda)$, $G \xleftarrow{\$} \G^*$, $\widetilde{G} \xleftarrow{\$} \widetilde{\G}^*$, $d, x_{1},  x_{2}  \xleftarrow{\$} \mathbb{Z}^{*}_{p}$,\\
~~~$X_{1} \leftarrow G^{x_{1}}$, $X_{2} \leftarrow G^{x_{2}}$, $\widetilde{H} \xleftarrow{\$} \widetilde{\G}$, $\widetilde{D} \leftarrow \widetilde{H}^{d}$, $\widetilde{U} \leftarrow \widetilde{H}^{x_{2} - dx_{1}}$.\\
~~~Return $\pp \leftarrow (\BGcal_{3}, G_{1}, G_{2}, X_{1}, X_{2}, \widetilde{H}, \widetilde{D}, \widetilde{U})$.\\

$\nOMSKGen(\pp):$\\
~~~$(y_{j,1},  y_{j,2})_{j \in [2]}  \xleftarrow{\$} (\mathbb{Z}^{*}_{p})^{4}$, $(\widetilde{V}_{j} \leftarrow \widetilde{H}^{y_{j, 2}} ( \widetilde{D}^{y_{j, 1}})^{-1})_{j \in [2]} , (\widetilde{Y}_{j, 1} \leftarrow \widetilde{G}^{y_{j,1}})_{j \in [2]}$.\\
~~~Return $(\pk, \sk) \leftarrow ((\widetilde{V}_{1}, \widetilde{V}_{2}), ((y_{j,1},  y_{j,2})_{j \in [2]}))$.\\

$\nOMSKVerify(\pk = (\widetilde{V}_{1}, \widetilde{V}_{2}), \sk = ((y_{j,1},  y_{j,2})_{j \in [2]})):$\\
~~~If $\widetilde{V}_{j} = \widetilde{H}^{y_{j, 2}} (\widetilde{D}^{y_{j, 1}})^{-1}$ for $j \in [2]$, return $1$.\\
~~~Otherwise return $0$.\\

$\nOMSKAgg(L_{n}=(\pk_{i} = (\widetilde{V}_{i, 1}, \widetilde{V}_{i, 2}))_{i \in [n]}):$\\
~~~$|n| > n_{max}(\lambda)$, return $\bot$.\\
~~~If there exists $(j, j')$ such that $j \neq j' \land \pk_{j} = \pk_{j'}$, return $\bot$.\\
~~~$\widetilde{\K}_{1} \leftarrow \prod_{i \in [n]}\widetilde{V}_{i, 1} , \widetilde{\K}_{2}  \leftarrow \prod_{i \in [n]}\widetilde{V}^{i}_{i, 2}$.\\
~~~Return $\apk = (\widetilde{\K}_{1}, \widetilde{\K}_{2})$.\\

$\nOMSSign(\sk_{n}= ((y_{n, j, 1},  y_{n, j, 2})_{j \in [2]}), L_{n-1}, m, \sigma_{n-1}= (A_{n-1}, B_{n-1}, C_{n-1})):$\\
~~~If $m = 0$, return $\bot$.\\
~~~If $L_{n-1} = \epsilon$ (i.e., $n-1=0$), $\sigma_{0} = (A_{0}, B_{0}, C_{0}) \leftarrow(G, X_{1}, X_{2})$.\\
~~~If $L_{n-1} \neq \epsilon$, if $\nSASSVerify(\nOMSKAgg(L_{n-1}), m, \sigma_{n-1}) = 0$, return $\bot$.\\
~~~If $L_{n-1} \neq \epsilon$, if there exists $(j, j')$ such that $j \neq j' \land \pk_{j} = \pk_{j'}$, return $\bot$. \\
~~~$r_{n} \xleftarrow{\$} \mathbb{Z}^{*}_{p}$, $A_{n} \leftarrow A^{r_{n}}_{n-1}$,\\
~~~$B_{n} \leftarrow (B_{n-1} A^{my_{n, 1, 1} + n y_{n, 2,1}}_{n-1})^{r_{n}}$, $C_{n} \leftarrow (C_{n-1} A^{my_{n, 1, 2} + n y_{n, 2,2}}_{n-1})^{r_{n}}$.\\
~~~Return $\sigma_{n} \leftarrow (A_{n}, B_{n}, C_{n})$.\\

$\nOMSSVerify(\apk = (\widetilde{\K}_{1}, \widetilde{\K}_{2}), m, \sigma=(A, B, C)):$\\
~~~If $m \neq 0  \land A \neq 1_{\G} \land e(A, \widetilde{U} \widetilde{\K}_{1}^{m} \widetilde{\K}_{2}) \cdot e(B, \widetilde{D}) = e(C, \widetilde{H})$, return $1$.\\
~~~Otherwise return $0$.\\
\hline
\end{tabular}
\caption{\small The ordered multi-signature scheme $\OMS_{\OursCK}$.}
\label{OMS_OursCK_const}
\end{figure}

\begin{lemma}\label{Th_OMS_SXDH}
If $\SAS_{\OursCK}$ satisfies the EUF-CMA security in the certified key model, then $\OMS_{\OursCK}$ with $\ell=2$ satisfies the EUF-CMA security in the certified key model.
\end{lemma}
We give the proof of Theorem \ref{Th_OMS_SXDH} in Appendix \ref{Subsec_Sec_OMS_Analysis}.
By combining Theorem \ref{Th_OMS_SXDH} with Corollary \ref{Coro_SXDH_SAS_Ours}, we obtain the following fact.

\begin{corollary}
If the SXDH assumption holds, then $\OMS_{\OursCK}$ satisfies EUF-CMA security in the certified key model.
\end{corollary}

\appendix

\section{Bilinear Groups}\label{Appen_BG}

A pairing group is a tuple $\BGcal = (p, \G, \widetilde{\G}, \G_T, e)$ where $\G$, $\widetilde{\G}$ and $\G_T$ are cyclic group of prime order $p$ and $e:\G \times \widetilde{\G} \rightarrow \G_T$ is an efficient computable bilinear map that satisfies the followings.
\begin{enumerate}
\item $X \in \G$, $\widetilde{Y} \in \widetilde{\G}$ and $a, b \in \Z_{p}$, $e(X^a, \widetilde{Y}^b) = e(X, \widetilde{Y})^{ab}$.
\item $G \in \G^{*}$, $\widetilde{G} \in \widetilde{\G}^{*}$, $e(G, \widetilde{G}) \neq 1_{\G_T}$.
\end{enumerate}
In this work, we use the type $3$ pairing group: $\G \neq \widetilde{\G}$ and there is no efficiently computable isomorphism $\psi:\widetilde{\G} \rightarrow \G$.

We introduce a bilinear group generator $\BG$ which takes as an input a security parameter $1^{\lambda}$ and returns the descriptions of an asymmetric pairing $\BGcal = (p, \G, \widetilde{\G}, \G_{T}, e)$ where $p$ is a $\lambda$-bits prime.
We represent a description of type $i$ pairing group as $\BGcal_{i}$.

\begin{assumption}[DDH Assumption on $\G$]\label{Assum_DDH_Gi}
Let $\BG$ be a bilinear group generator and $\A$ be a PPT algorithm.
The decisional Diffie-Hellman (DDH) assumption on $\G$ holds for $\BG$ if for any PPT adversary $\A$, the following advantage 

\begin{equation*}
\begin{split}
&\Adv^{\DDH_{\G}}_{\BG, \A}(\lambda) := \left|\Pr[1 \leftarrow \A(\BGcal_{3}, G, \widetilde{G}, S, T, Z_{1})] - \Pr[1 \leftarrow \A(\BGcal_{3}, G, \widetilde{G}, S, T, Z_{0})] 
\right|
\end{split} 
\end{equation*}
is negligible in $\lambda$ where $\BGcal_{3} \leftarrow \BG(1^{\lambda}), s, t, z \xleftarrow{\$} \Z_{p}, G \xleftarrow{\$} \G^{*}, \widetilde{G} \xleftarrow{\$} \widetilde{\G}^{*}, S \leftarrow G^{s},  T \leftarrow G^{t}, Z_{b} \leftarrow G^{st + bz}$.
\end{assumption}

The dual of the above assumption is the DDH assumption on $\widetilde{\G}$ for $\BG$, which is defined by changing from $(S, T, Z_{b})$ to $(\widetilde{S}, \widetilde{T}, \widetilde{Z_{b}})$ in Definition \ref{Assum_DDH_Gi} where $ \widetilde{S} \leftarrow  \widetilde{G}^{s},   \widetilde{T} \leftarrow  \widetilde{G}^{t},  \widetilde{Z}_{b} \leftarrow  \widetilde{G}^{st + bz}$.

\begin{assumption}[SXDH Assumption \cite{BGMM05}]\label{Assum_SXDH}
Let $\BG$ be a bilinear group generator.
The symmetric external Diffie-Hellman (SXDH)  assumption holds for $\BG$ if the DDH assumption holds both $\G$ and $\widetilde{\G}$.
\end{assumption}

\begin{assumption}[DBP Assumption \cite{Gro09}]\label{Assum_SXDH}
Let $\BG$ be a bilinear group generator and $\A$ be a PPT algorithm.
The double pairing (DBP) assumption on $\widetilde{\G}$ holds if for any PPT adversary $\A$, the following advantage 
\begin{equation*}
\begin{split}
&\Adv^{\DBP_{\G}}_{\BG, \A}(\lambda)   \\
&~~~~:= \left|
\Pr\left[
\begin{split}
&e(E, \widetilde{H}) \cdot e(F, \widetilde{D}) = 1_{\G_{T}} \\
&\land (E, F) \neq (1_{\G}, 1_{\G})
\end{split}  
\middle|
\begin{split}
&\BGcal_{3} \leftarrow \BG(1^{\lambda}), G \xleftarrow{\$} \G^{*}, \widetilde{G} \xleftarrow{\$} \widetilde{\G}^{*}, \\
&\widetilde{H}, \widetilde{D} \xleftarrow{\$} \widetilde{\G}^{*}, (E, F) \leftarrow \A(\BGcal_{3}, G, \widetilde{G}, \widetilde{H}, \widetilde{D})
\end{split}  
\right]
\right|
\end{split} 
\end{equation*}
is negligible in $\lambda$.
\end{assumption}

The dual of the above assumption is the DBP assumption on $\widetilde{\G}$ for $\BG$, which is defined by changing from $(\widetilde{H}, \widetilde{D}, E, F)$ to $(H, D, \widetilde{E}, \widetilde{F})$ in Definition \ref{Assum_DDH_Gi} where $H, D \in \G$ and $\widetilde{E}, \widetilde{F} \in \widetilde{\G}$.

\begin{lemma}[\cite{AFGHO16}]\label{Lem_from_SXDH_to_DBP}
If the SXDH assumption holds, the DBP assumptions on $\G$ and on $\widetilde{\G}$ hold.
\end{lemma}

\section{Digital Signatures (DS)}\label{Appen_DS_Def}
We review the definition of a digital signature scheme and its security.

\begin{definition}[Digital Signature Scheme]
A digital signature scheme $\DS$ consists of the following tuple of algorithms $(\nDSSetup, \nDSKGen, \allowbreak \nDSSign, \nDSVerify)$.
\begin{itemize}
\item $\nDSSetup (1^{\lambda}):$ A setup algorithm takes as an input a security parameter $1^{\lambda}$. It returns the public parameter $\pp$.
We assume that $\pp$ defines a message space and represents this space by $\mathcal{M}$.
We omit a public parameter $\pp$ in the input of all algorithms except for $\nDSKGen$.

\item $\nDSKGen (\pp):$ A key generation algorithm takes as an input a public parameter $\pp$. It returns a public key $\pk$ and a secret key $\sk$.

\item $\nDSSign (\sk, m):$ A signing algorithm takes as an input a secret key $\sk$ and a message $m$. It returns a signature $\sigma$.

\item $\nDSVerify (\pk, m, \sigma):$ A verification algorithm takes as an input a public key $\pk$, a message $m$, and a signature $\sigma$.
It returns a bit $b \in  \{0, 1\}$.
\end{itemize}
\end{definition}
\paragraph{\bf Correctness.}
$\DS$ satisfies correctness if $\forall \lambda \in \N$, $\pp \leftarrow \nDSSetup (1^{\lambda})$, $\forall m \in \mathcal{M}_{\pp}$, $(\pk, \sk) \leftarrow \nDSKGen(\pp)$, and $\sigma \leftarrow \nDSSign(\sk, m)$, $\nDSVerify(\pk, m, \sigma) = 1$ holds.

\begin{definition}[EUF-CMA Security]
Let $\DS$ be a digital signature scheme and $\A$ be a PPT adversary.
The existentially unforgeable under chosen message attacks (EUF-CMA) security is defined by the following EUF-CMA game $\sfGame^{\EUFCMA}_{\DS, \A}(1^{\lambda})$ between the challenger $\C$ and an adversary $\A$.

\begin{itemize}
\item {\bf Initial setup:}
$\C$ initializes a set $S^{\Sign} \leftarrow \{\}$, runs $\pp \leftarrow \nDSSetup (1^{\lambda})$, $(\pk^{*}, \sk^{*}) \leftarrow \nDSKGen(\pp)$, and sends $(\pp, \pk^{*})$ to $\A$.
\item $\A$ makes signing queries polynomially many times.
\begin{itemize}
\item  {\bf Signing query:}
For an signing query on $m$, $\C$ updates $S^{\Sign} \leftarrow S \cup \{m\}$, runs $\sigma \leftarrow \nDSSign (\sk^{*}, m)$, and returns $\sigma$ to $\A$.
\end{itemize}
{\bf End of the game:}  
$\A$ finally outputs a forgery $(m^{*}, \sigma^{*})$ to $\C$.\\
If $m^{*} \notin S^{\Sign} \land \nDSVerify (\pk^{*}, m^{*}, \sigma^{*}) = 1$, return $1$. Otherwise, return $0$.
\end{itemize}

The advantage of an adversary $\A$ for the game is defined by $\Adv^{\EUFCMA}_{\DS, \A}(\lambda):= \Pr[\sfGame^{\EUFCMA}_{\DS, \A}(1^{\lambda}) \allowbreak \Rightarrow 1]$.
$\DS$ satisfies the EUF-CMA security if for any PPT adversary $\A$, $\Adv^{\EUFCMA}_{\DS, \A}(\lambda)$ is negligible in $\lambda$.
\end{definition}

\section{Proof of Theorem \ref{Th_OMS_SXDH}}\label{Subsec_Sec_OMS_Analysis}
We give a security proof for Theorem \ref{Th_OMS_SXDH}.

\begin{proof}
Let $\A$ be a PPT adversary for the EUF-CMA security of $\OMS_{\OursCK}$ in the certified key model.
We give a reduction algorithm $\B$ for the EUF-CMA security of $\SAS_{\OursCK}$ with $\ell = 2$ as follows.
\begin{itemize}
\item {\bf Initial setup:}
$\B$ takes as an input an instance $(\pp, \pk^{*}) = ((\BGcal_{3}, G, \widetilde{G}, X_{1}, \allowbreak X_{2}, \widetilde{H}, \widetilde{D}, \widetilde{U}), (\widetilde{V}_{1}, \widetilde{V}_{2}))$ for the EUF-CMA security game of $\SAS_{\OursCK}$ with $\ell=2$.
$\B$ initialize sets $S^{\Cert} \leftarrow \{\}$, $S^{\Sign} \leftarrow \{\}$.
Then, $\B$ sends an instance $(\pp', \pk') \leftarrow (\pp, \pk^{*})$ as an input.

\item {\bf Key registration query:}
For a key registration query on $(\pk, \sk)$, $\B$ checks $\nOMSKVerify(\pk, \sk) \allowbreak = 1$.
If this condition holds, $\B$ updates $S^{\Cert} \leftarrow S^{\Cert} \cup \{(\pk, \sk)\}$, makes a key registration query for the EUF-CMA security game of $\SAS_{\OursCK}$ on $(\pk, \sk)$, and returns $\accept$.
Otherwise, $\B$ returns $\reject$.

\item {\bf Signing query:}
For a signing query on $(L_{n-1} = (\pk_{1}, \dots, \pk_{n-1}), m, \sigma_{n-1})$, if $L_{n-1}  \neq  \epsilon$ (i.e., $n-1 \geq 1$), $\B$ checks that the following conditions:
\begin{enumerate}
\item $\pk^{*} \notin L_{n-1} ;$
\item $\pk_{i} \in Q^{\Cert}$ for $i \in [n-1] ;$
\item $\nOMSSVerify(\nOMSKAgg(L_{n-1}), m, \sigma_{n-1}) = 1$.
\end{enumerate}
If at least one of the above conditions does not hold, $\B$ returns $\bot$.
Then, $\B$ makes a signing query for the EUF-CMA security game of $\SAS_{\OursCK}$ on $(L_{n-1}, (m_{i} = (m, i))_{i \in [n-1]}, m_{n} = (m, n), \sigma_{n-1})$.
Then $\B$ obtains a signature $\sigma_{n} = (A_{n}, B_{n}, C_{n})$. 
$\B$ updates $S^{\Sign} \leftarrow S^{\Sign} \cup \{(m, n)\}$ and returns $\sigma_{n}$ to $\A$.

\item {\bf End of the game:}
After receiving the forgery $(L^{*}_{n^{*}}=(\pk^{*}_{1}, \dots, \pk^{*}_{n^{*}}), m^{*}, \allowbreak \sigma^{*}_{n^{*}} = (A^{*}_{n^{*}}, B^{*}_{n^{*}}, C^{*}_{n^{*}}))$ from $\A$, 
$\B$ checks the following conditions hold.
\begin{enumerate}
\item $n^{*} \leq n_{max} \land \nOMSSVerify(\nOMSKAgg(L^{*}_{n^{*}}), m^{*}, \sigma^{*}_{n^{*}}) = 1;$
\item There exists $i^{*} \in [n^{*}]$ such that $\pk^{*}_{i^{*}} = \pk^{*} \land (m^{*}, i^{*}) \notin S^{\Sign};$
\item $\forall j \in [n^{*}]$ such that $\pk_{j} \neq \pk^{*}$, $\pk_{j} \in Q^{\Cert}.$
\end{enumerate}
Then, $\B$ outputs $(L^{*}_{n^{*}}, (m^{*}_{i} = (m^{*}, i))_{i \in [n^{*}]}, \sigma^{*}_{n^{*}})$ as a forgery.
\end{itemize}

Clearly, $\B$ simulates the EUF-CMA game of $\OMS_{\OursCK}$.
We confirm that if $\A$ outputs a valid forgery $(L^{*}_{n^{*}}=(\pk^{*}_{1}, \dots, \pk^{*}_{n^{*}}), m^{*}, \allowbreak \sigma^{*}_{n^{*}} = (A^{*}_{n^{*}}, B^{*}_{n^{*}}, C^{*}_{n^{*}}))$ for the EUF-CMA game of $\OMS_{\OursCK}$, $\B$ outputs a valid forgery for the EUF-CMA game of $\SAS_{\OursCK}$.
Since a forgery is valid, $(L^{*}_{n^{*}}, (m^{*}_{i} = (m^{*}, i))_{i \in [n^{*}]})$ is not queried for the signing oracle of the EUF-CMA game of $\SAS_{\OursCK}$. 
Since $\nOMSKAgg(L^{*}_{n^{*}}) \allowbreak = (\widetilde{\K}_{1}, \widetilde{\K}_{2}) = (\prod_{i \in [n]}\widetilde{V}_{i, 1} ,  \prod_{i \in [n]}\widetilde{V}^{i}_{i, 2})$ and $\nOMSSVerify(\nOMSKAgg(L^{*}_{n^{*}}), m^{*}, \sigma^{*}_{n^{*}}) \allowbreak = 1$ hold, $e(A^{*}, \widetilde{U} \prod_{i \in [n]}\widetilde{V}_{i, 1}^{m} \prod_{i \in [n]}\widetilde{V}^{i}_{i, 2}) \cdot e(B^{*}, \widetilde{D}) = e(C^{*}, \widetilde{H})$ holds where $\pk^{*}_{i} = (\widetilde{V}_{i, 1}, \widetilde{V}_{i, 2})$ for $i \in [n^{*}]$.
We see that $(L^{*}_{n^{*}}, (m^{*}_{i} = (m^{*}, i))_{i \in [n^{*}]}, \sigma^{*}_{n^{*}})$ is a valid forgery for the EUF-CMA game of $\SAS_{\OursCK}$ with $\ell=2$.
Hence, we conclude Theorem~\ref{Th_OMS_SXDH}.
\qed
\end{proof}

\section*{Acknowledgement}
A part of this work was supported by JSPS KAKENHI JP24H00071, JP23K16841 and JST CREST JPMJCR2113, and JST K Program JPMJKP24U2.

\bibliographystyle{abbrvurl}
\bibliography{OMS_PS}

\newpage
\setcounter{tocdepth}{2}
\tableofcontents

\end{document}